\date{\today}
\documentclass{amsproc}

\usepackage{amssymb}

\usepackage{graphicx}

\usepackage[cmtip,all]{xy}

\usepackage[all]{xy}
\usepackage{amsthm}
\usepackage{amscd}
\usepackage{amsmath}
\usepackage{mathtools}
\usepackage{amsfonts}
\usepackage{amssymb}
\usepackage{color}
\usepackage{tikz}
\usetikzlibrary{cd}
\usepackage{stmaryrd}
\usepackage{graphicx}%
\usepackage{hyperref}
\hypersetup{
    bookmarks=true,         
    unicode=false,          
    pdftoolbar=true,        
    pdfmenubar=true,        
    pdffitwindow=false,     
    pdfstartview={FitH},    
    pdftitle={My title},    
    pdfauthor={Author},     
    pdfsubject={Subject},   
    pdfcreator={Creator},   
    pdfproducer={Producer}, 
    pdfkeywords={keyword1, key2, key3}, 
    pdfnewwindow=true,      
    colorlinks=false,       
    linkcolor=red,          
    citecolor=green,        
    filecolor=magenta,      
    urlcolor=cyan           
}

\newcommand*{\emptycomment}[1]{}

\newtheorem{theorem}{Theorem}[section]
\newtheorem{lemma}[theorem]{Lemma}

\theoremstyle{definition}
\newtheorem{definition}[theorem]{Definition}
\newtheorem{example}[theorem]{Example}
\newtheorem{proposition}[theorem]{Proposition}
\newtheorem{corollary}[theorem]{Corollary}

\newcommand{\git}{/\!\!/}
\newcommand{\refl}{\text{\textnormal{ref}}}
\newcommand{\Irr}{\text{\textnormal{Irr}}}
\newcommand{\Stab}{\text{\textnormal{Stab}}}
\newcommand{\cpt}{\text{\textnormal{cpt}}}

\theoremstyle{remark}
\newtheorem{remark}[theorem]{Remark}

\numberwithin{equation}{section}

\DeclareMathOperator {\Map}{Map}

\DeclareMathOperator {\QR} {QEllR}
\DeclareMathOperator {\pt} {pt}
\DeclareMathOperator {\id} {Id}

\def \Z {\mathbb{Z}}
\def \C {\mathbb{C}}
\def \R {\mathbb{R}}

\def \T {\mathbb{T}}

\begin{document}
\title{ Quasi-elliptic cohomology of 4-spheres} 

\author{Zhen Huan}

\address{Zhen Huan, Center for Mathematical Sciences,
Huazhong University of Science and Technology, Hubei 430074, China} \curraddr{}
\email{2019010151@hust.edu.cn}

\date{\today}


\begin{abstract}    
Quasi-elliptic cohomology is conjectured in \cite{SatiSchreiber2022}
as a particularly suitable approximation to equivariant
$4$-th $Cohomotopy$, which classifies the
charges carried by M-branes in M-theory in a way that is analogous to
the traditional idea that complex K-theory classifies the charges of D-branes in
string theory. 
In this paper we compute quasi-elliptic cohomology of 4-spheres under the action by some finite subgroups that are the most interesting isotropy groups where the M5-branes may sit.
\end{abstract}

\maketitle

\tableofcontents

\section{Introduction}

In this paper we compute Real and complex quasi-elliptic cohomology of 4-spheres under specific action of some finite subgroups of $\mathrm{Spin}(5)$, which aims to give an approximation to the equivariant unstable 4th Cohomotopy, which is especially difficult to compute. Cohomotopy theory is conjectured to be the actual cohomology theory of relevance for classifying brane charges in M-theory. 

To interpret the relation between the computation and cohomotopy, we start the story by classifying spaces for cohomology theories.
For a given cohomology theory $E^{\ast}(-)$ with classifying space $E$, we have, for any good enough space $X$, \[ E^0(X) = \pi_0 \Map(X, E). \]
Here we can regard a map $X\rightarrow E$ as a "cocycle" for the $E$-cohomology, and a homotopy between such maps as a "boundary" in $E$-cohomology. Generally,  the classifying space of an abelian cohomology theory is its spectrum at level $0$. A classical example is  complex  topological $K$-theory
$K(-)$, whose classifying space can be taken to be $KU = BU \times \Z$. 

In addition, instead of using the whole spectrum of $E$, with only the classifying space we can define a generalized  non-abelian cohomology theory \[ E(X) := \pi_0 \Map(X, E)\] which makes good sense. One issue is that computing such cohomology theories is generally difficult. 
One method is approximating the cohomology theory $E$ by another one $E'$, which is better understood and easier to compute. The method is clearly possible whenever there is a map of classifying spaces $ E \longrightarrow E'$ because it induces evidently a cohomology operation 
\[ E(-) \longrightarrow E'(-),\]
which provides an image of the less-understood $E$-cohomology in the better-understood $E'$-cohomology. 

The archetypical example here is the Chern-Dold character map \cite{Dold_Adams_Shepherd_1972} \cite{FSS_char_nonabelian_cohomology_2023}, which approximates any generalized cohomology theory by a rational
cohomology theory. For instance,  the ordinary Chern character on $K(-)$
  \[ K(-) \longrightarrow H^{ev}(-; \mathbb{Q}) \] with $H^{ev}(-; \mathbb{Q}):= \prod\limits_{n\in\mathbb{N}} H^{2n}(-; \mathbb{Q})$, 
is represented by a map of classifying spaces 
 \[  BU \times  \Z  \longrightarrow   \prod\limits_{n \in \mathbb{N}} K(\mathbb{Q}, 2n).  \] 
This map of classifying spaces is itself a cocycle in the rational cohomology of the classifying space $BU\times \Z$. In other words, the Chern character itself can be viewed as an element in \[H^{ev}(BU\times \Z; \mathbb{Q}).\] 

Generally, a map of classifying spaces $E\longrightarrow E'$, inducing a cohomology operation $E(-) \longrightarrow E'(-) $,
 is itself a cocycle in the $E'$-cohomology $E'(E)$ of the classifying space $E$. 
Thus, in order to understand $E$-cohomology, we may try to understand  the $E'$-cohomology of its classifying space $E$ for suitable alternative
cohomology theories $E'$. 

Now we consider the cohomology theory, the $n$-th cohomotopy theory \[nCohomotopy(-),\]  whose classifying space is an $n$-sphere $S^n$.
Each  cocycle in the $E'$-cohomology $E'(S^n)$ is represented by a map $S^n \longrightarrow E'$. From it, 
 we get a cohomology operation
\[ nCohomotopy(-) \longrightarrow E'(-) ,\]
which provides us images of $nCohomotopy$ in $E'$-cohomology similarly to
how the Chern character provides images of $K$-cohomology in ordinary
rational cohomology.

It is suggested by \textit{Hypothesis H} \cite{FSS2024} \cite{SS_2020_103775} \cite{SS:2021uhj} that, specifically, $\mathrm{Spin}(5)$-twisted equivariant unstable $4Cohomotopy$ classifies the
charges carried by M-branes in M-theory in a way that is analogous to
the traditional idea that $K(-)$ classifies the charges of D-branes in
string theory.
Therefore, it's essential to compute the $4Cohomotopy$ of spacetime domains
relevant in M-theory. This can be hard, in particular once we remember
that all of these need to be done in twisted equivariant generality.
Thus, we apply the idea to approximate $4Cohomotopy$ of spacetime by using the cocycles $$S^4\longrightarrow E'$$
in $E'(S^4)$ for some suitable cohomology theory $E'$. Instead of $4Cohomotopy$ itself, we will study the image of the corresponding 
cohomology operation \[ 4Cohomotopy(-) \longrightarrow E'(-).\] Some information of the actual $4Cohomotopy$ may be lost but what they retain can still be valuable and is expected to be better understandable.

Specifically, the classifying spaces for equivariant $4Cohomotopy$ are
orbifolds $S^4\git G$ of the $4$-sphere acted by a group $G$, i.e. the orbifolds of
representation 4-spheres.  Hence the elements of the
$G$-equivariant $E'$-cohomology  $E'_G(S^4)$ 
serve, in the above way, as "generalized equivariant characters" on
equivariant $4Cohomotopy$, namely as equivariant cohomology operation
\[4Cohomotopy_G(-) \longrightarrow  E'_G(-).\]

As conjectured in \cite{SatiSchreiber2022}, the choice 
\[  E'_G(-)  :=  QEll_G(-) \]
should be a particularly suitable approximation to equivariant
$4Cohomotopy$ for the purpose of computing M-brane charge. 
One motivation for this is that the Witten elliptic genus,  which was
originally discussed for string \cite{Witten1988},  actually makes sense for
M5-branes \cite{ISM2004} \cite{KS2005} \cite{DGY_2007} \cite{GPP2021_4TMF} \cite{AHKRW_2015},
so that one should expect that it is actually part of the charges
carried by M5-branes. But these charges should also be in $Cohomotopy$,
and hence, it is conjectured in \cite{SatiSchreiber2022} that there is a useful approximation of
$4Cohomotopy$ by elliptic cohomology, and specifically by
quasi-elliptic cohomology.

This is the motivation for computing the quasi-elliptic cohomology for
representation 4-spheres. Moreover, as indicated in \cite{FSS_char_nonabelian_cohomology_2023}, the particular choice of
equivariance groups $G$  as finite subgroups of $\mathrm{Spin}(5)$ comes from the
fact that these are the most interesting isotropy groups for the
orbifolds on which these M5-branes may sit. We describe the interesting groups and their action on 4-spheres below.

The space $\mathbb{H}$ of quaternions is isomorphic to $\mathbb{R}^4$ as a real vector space. In addition,  the group of the unit quaternions is isomorphic to the special unitary group $SU(2)$, which is isomorphic to $\mathrm{Spin}(3)$. It can be identified with a subgroup of $\mathrm{Spin} (5)$ via the composition
\[ \mathrm{Spin}(3) \buildrel{p_1}\over\hookrightarrow \mathrm{Spin}(3)\times \mathrm{Spin}(3) \cong \mathrm{Spin}(4) \hookrightarrow \mathrm{Spin}(5)\] where the first homomorphism is the inclusion into the first factor.  
Under quaternion multiplication, there are two choices of group action by $\mathbb{H}$ on $\mathbb{R}^4$ that we are especially interested in.

  \begin{equation}\label{spin(4)_Act_12}
    \begin{tikzcd}[
      row sep=0pt,
      column sep=0pt
    ]
      \mathrm{Spin}(4)
      \ar[
        r,
        phantom,
        "{\simeq}"
      ]
      &[8pt]
      \mathrm{Spin}(3)
      \times
      \mathrm{Spin}(3)
      \ar[rr, ->>]
      &&
      \mathrm{SO}(\mathbb{H})
      \ar[
        r,
        phantom,
        "{ \simeq }"
      ]
      &[-12pt]
      \mathrm{SO}(4)
      \\
      &
      (e_1, 1)
      &\mapsto&
      \big(
        q
        \,\mapsto\,
        e_1 \cdot q 
      \big)
        \\
       &
      (e_1, e_1)
      &\mapsto&
      \big(
        q
        \,\mapsto\,
        e_1 \cdot q \cdot e_1^\ast
      \big)
    \end{tikzcd}
  \end{equation}

The group action can extend to $S^4$ by keeping the north pole and the south pole fixed. In \cite[Section 6]{huan2020_v2} we compute complex quasi-elliptic cohomology of $S^4$ under the first group action in \eqref{spin(4)_Act_12}. In Section \ref{Real_spin3_s4} we compute the Real quasi-elliptic cohomology for that. Moreover, in Section \ref{prod_rc_qell}, we compute some examples of complex and Real quasi-elliptic cohomology of $S^4$ under the second group action in \eqref{spin(4)_Act_12}.

In the appendix, we give some corollaries of the decomposition formula for complex equivariant $K$-theories in \cite{ngel2017EquivariantCB} and 
the Mackey decomposition formula for Freed-Moore $K$-theories in \cite{huanyoung2022}. They are used in the computation in Section \ref{Real_spin3_s4} and
\ref{prod_rc_qell} respectively.

In addition, before we present the computation of quasi-elliptic cohomology, we review in Section \ref{qell} and \ref{QEllR_rev} quasi-elliptic cohomology and twisted Real quasi-elliptic cohomology respectively.

\section*{Acknowledgments} 

This research is based upon the work  supported by  the General Program of the National Natural Science Foundation of China (Grant No. 12371068)  for the project 
"Quasi-elliptic cohomology and its application in topology and mathematical physics", the National Science Foundation under Grant Number DMS 1641020,  and the research funding from Huazhong University of Science and Technology.

The author thanks Hisham Sati and Urs Schreiber for suggesting the author compute quasi-elliptic cohomology of spheres and helpful discussion, and thanks Center for Quantum and Topological Systems at New York University Abu Dhabi for hospitality and support. In addition, the author thanks Beijing International Center for Mathematical Research and Peking University for hospitality and support. Part of this work was done during the author's visit at BICMR.

\section{Quasi-elliptic cohomology}\label{qell}

In this section we recall the definition of quasi-elliptic
cohomology in term of equivariant K-theory and state the
conclusions that we need in this paper. For more details on quasi-elliptic
cohomology, please refer to \cite{Huan2018a}. 

Let $G$ be a compact Lie group and  $X$ a $G$-space. Let $G^{tors}\subseteq G$ denote the set of
torsion elements of $G$. For any $g\in G^{tors}$, the fixed point
space $X^{g}$ is a $C_G(g)-$space where $C_G(g)$ is the centralizer $\{h\in G\mid h g= g h\}$. This group action can be extended to that by the group
\[\Lambda_G(g):=C_G(g)\times \mathbb{R}/\langle(g, -1)\rangle,\] which is given explicitly by 
\begin{equation} [h, t]\cdot x:=h \cdot x,  \label{lambda_act_def}\end{equation} for any $[h, t] \in \Lambda_G(g)$ and $x\in X^{g}$.

To give a complete description of the loop groupoid $\Lambda(X/\!\!/G)$, we need the following definitions.

\begin{definition}
\begin{enumerate}\item Let $g$, $g'$ be two elements in $G$. Define $C_G(g, g')$ to be the set 
$\{ h\in G\mid g'h=hg\}.$

\item Let $\Lambda_G(g, g')$ denote
the quotient of $C_G(g, g')\times \mathbb{R}/l\mathbb{Z}$ under
the equivalence $$(\alpha, t)\sim (g'\alpha, t-1)=(\alpha g,
t-1),$$ where $l$ is the order of $g$ in $G$. \end{enumerate} \end{definition}

\begin{definition} Define $\Lambda(X/\!\!/G)$ to be the groupoid with
\begin{itemize}
\item \textbf{objects}: the space $\coprod\limits_{g\in G^{tors}}X^{g}$
\item 
\textbf{morphisms}: the space
$$\coprod\limits_{g, g'\in G^{tors}}\Lambda_G(g, g')\times X^g.$$ \end{itemize}
For an object $x\in X^g$, the morphism
$([\alpha, t], x)\in \Lambda_G(g, g')\times X^g$ is an arrow from $x$ to $\alpha \cdot x\in X^{g'}.$ The composition of the morphisms is
defined by
\begin{equation}([\alpha_1, t_1], \alpha_2\cdot x) \circ ([\alpha_2,
t_2], x) = ([\alpha_1\alpha_2, t_1+t_2], x).\end{equation} 
Let $\mathbb{T}$ denote the circle group $\mathbb{R}/\mathbb{Z}$. We have a
homomorphism of orbifolds
$$\pi: \Lambda(X/\!\!/G)\longrightarrow B\mathbb{T}$$ sending all the objects to the single object in $B\mathbb{T}$, and
a morphism $([\alpha,t], x)$ to $e^{2\pi it}$ in $\mathbb{T}$.

\end{definition}

\begin{definition} The quasi-elliptic cohomology $QEll^*_G(X)$ is defined to be
$K^*_{orb}(\Lambda(X/\!\!/G))$. \label{defqell1}\end{definition}

The groupoid $\Lambda(X\git G)$ is equivalent to the disjoint union of action groupoids \begin{equation}\coprod_{g\in \pi_0(G^{tors} \git G)} X^g \git \Lambda_G(g) \label{lambda_decomp} \end{equation} where
$G^{tors}\git G$ is the conjugation quotient groupoid.   Thus,
we can unravel Definition \ref{defqell1} and
express it via equivariant K-theory.


\begin{definition}
\begin{equation}\label{defqell} QEll^*_G(X):=\prod_{g\in \pi_0(
G^{tors} \git G) }K^{\ast}_{\Lambda_G(g)}(X^{g})=\bigg(\prod_{g\in
G^{tors}}K^*_{\Lambda_G(g)}(X^{g})\bigg)^G.\end{equation} 
\end{definition}

Consider the composition
$$\mathbb{Z}[q^{\pm}]=K_{\mathbb{T}}(\mbox{pt})\buildrel{\pi^*}\over\longrightarrow K_{\Lambda_G(g)}(\mbox{pt})\longrightarrow
K_{\Lambda_G(g)}(X)$$ where $\pi: \Lambda_G(g)\longrightarrow
\mathbb{T}$ is the projection $[a, t]\mapsto e^{2\pi i t}$ and the
second map is defined via the collapsing map $X\longrightarrow \mbox{pt}$. Via it, 
$QEll_G^{\ast}(X)$ is naturally a
$\mathbb{Z}[q^{\pm}]-$algebra. 

\begin{proposition}The relation between quasi-elliptic cohomology and equivariant 
Tate K-theory $K^{\ast}_{Tate}(- \git G)$ is
\begin{equation}QEll^{\ast}_G(X) \otimes_{\mathbb{Z}[q^{\pm}]}\mathbb{Z}((q)) \cong K^{\ast}_{Tate}(X\git G).
\label{tateqellequiv}\end{equation}\end{proposition} This is the main reason why the theory is called quasi-elliptic cohomology.

In addition, we give an example computing quasi-elliptic cohomology, which is \cite[Example 3.3]{Huan2018a}. The conclusions in Example \ref{ex3.3:huan2018} are applied in the computation of Section \ref{Real_spin3_s4} and Section \ref{prod_rc_qell}.
\begin{example}[$G=\mathbb{Z}/N$] \label{ex3.3:huan2018}
Let $G=\mathbb{Z}/N$ for $N\geq 1$, and let $\sigma\in
G$. Given an integer $k\in\mathbb{Z}$ which projects to
$\sigma\in\mathbb{Z}/N$, let $x_k$ denote the
representation of $\Lambda_G(\sigma)$ defined by
\begin{equation}\begin{CD}\Lambda_{G}(\sigma)=(\mathbb{Z}\times\mathbb{R})/(\mathbb{Z}(N,0)+\mathbb{Z}(k,1))
@>{[a,t]\mapsto[(kt-a)/N]}>> \mathbb{R}/\mathbb{Z}=\mathbb{T}
@>{q}>> U(1).\end{CD}\label{xk}\end{equation} $R\Lambda_G(\sigma)$
is isomorphic to the ring $\mathbb{Z}[q^{\pm}, x_k]/(x^N_k-q^k)$.
\bigskip

For any finite abelian group 
$G=\mathbb{Z}/N_1\times\mathbb{Z}/N_2\times\cdots\times\mathbb{Z}/N_m$,
let $\sigma=(k_1, k_2, \cdots k_n)\in G.$ We have
$$\Lambda_G(\sigma)\cong\Lambda_{\mathbb{Z}/N_1}(k_1)\times_{\mathbb{T}}\cdots\times_{\mathbb{T}}\Lambda_{\mathbb{Z}/N_m}(k_m).$$ Then
\begin{align*}R\Lambda_G(\sigma)&\cong
R\Lambda_{\mathbb{Z}/N_1}(k_1)\otimes_{\mathbb{Z}[q^{\pm}]}\cdots\otimes_{\mathbb{Z}[q^{\pm}]}R\Lambda_{\mathbb{Z}/N_m}(k_m)\\
&\cong\mathbb{Z}[q^{\pm}, x_{k_1}, x_{k_2},\cdots
x_{k_m}]/(x^{N_1}_{k_1}-q^{k_1},x^{N_2}_{k_2}-q^{k_2}, \cdots
x^{N_m}_{k_m}-q^{k_m})\end{align*} where all the $x_{k_j}$'s are
defined as $x_k$ in (\ref{xk}).\label{ppex}
\end{example}

\section{Twisted Real quasi-elliptic cohomology} \label{QEllR_rev}

In this section, we review the definition and properties of twisted Real quasi-elliptic cohomology. For more details, please refer to \cite{huanyoung2022}.

\begin{definition}
Let $G$ be a finite group. 
A \emph{$\Z/2$-graded group} is a group homomorphism $\pi: \hat{G} \rightarrow \Z/2$. The \emph{ungraded group} of $\hat{G}$ is $G = \ker \pi$. When $\pi$ is non-trivial,  $\hat{G}$ is called a \emph{Real structure} on $G$. The group $\hat{G}$ acts on $G$ by Real conjugation, \[\varsigma \cdot g = \varsigma g^{\pi(\varsigma)} \varsigma^{-1},\] $g \in G, \;\varsigma \in \hat{G}$. The Real centralizer of $g\in G$ is
\[
C^R_{\hat{G}}(g) = \{ \varsigma \in \hat{G} \mid \varsigma g ^{\pi(\varsigma)}\varsigma^{-1} = g \}.
\]
The group $C^R_{\hat{G}}(g)$ is $\Z/2$-graded with ungraded group the centralizer $C_{G}(g)$. \end{definition}

\begin{example}
The terminal $\Z/2$-graded group is $\id: \Z/2 \rightarrow \Z/2$ and is denoted simply by $\Z/2$. If $\Z/2$ acts on a group $\hat{H}$, then so does any $\Z/2$-graded group $\hat{G}$ and the resulting semi-direct product $\hat{H} \rtimes_{\pi} \hat{G}$ is naturally $\Z/2$-graded.
\end{example}

\begin{example}
The dihedral group $D_{2n}$ \[\langle r, s \mid r^n=1, s^2=1, (sr)^2 =1 \rangle. \] is a Real structure on $\Z/n$. The subgroup $\langle r \rangle\cong \Z/n$ is a normal subgroup of $D_{2n}$ and we have the short exact sequence\[ 1\longrightarrow \Z/n \longrightarrow D_{2n} \longrightarrow \Z/2 \longrightarrow 1\] with a generator of $\Z/n$ mapped to the rotation $r$. 
\end{example}

\begin{example}
    As computed in \cite[Example 1.8]{huanyoung2022}, the Real representation ring $RR(\Z/n)$ w.r.t. to the Real structure $D_{2n}$ is isomorphic to complex representation ring $R(\Z/n) \cong \Z[\zeta]/\langle \zeta^n-1 \rangle.$ 
\end{example}

\begin{example}
    For any $g\in G$, the Real centralizer $C^R_{\hat{G}}(g)$ is $\Z/2$-graded with ungraded group the centralizer $C_{G}(g)$. It is a Real Structure on $C_G(g)$.
    
In addition, the element $(-1,g) \in \R \rtimes_{\pi} C_{\hat{G}}^R(g)$ is Real central and so generates a normal subgroup isomorphic to $\Z$. This leads to the definition of the  \emph{Real enhanced centralizer} of $g$.
\[
\Lambda^R_{\hat{G}}(g)
:=
\left( \R \rtimes_{\pi} C_{\hat{G}}^R(g) \right) \slash \langle (-1,g) \rangle.
\] It is a Real structure on the group $\Lambda_G(g) $.

\end{example}

The set of connected components $\pi_0(G \git G)$ of the conjugation quotient groupoid is the set of conjugacy classes of $G$. Given a Real structure $\hat{G}$, Real conjugation defines an involution of $\pi_0( G \git  G)$. This defines a partition
\begin{equation}
\label{eq:conjClassesDecomp}
\pi_0(G \git G)
=
\pi_0(G \git G)_{-1} \sqcup \pi_0(G \git G)_{+1}
\end{equation}
with $\pi_0(G \git G)_{-1}$ the fixed point set of the involution. The conjugacy class of $g \in G$ is fixed by the involution if and only if $C_{\hat{G}}^R(g) \setminus C_{G}(g) \neq \varnothing$. The set $\pi_0(G \git_R \hat{G})$ of Real conjugacy classes of $G$ inherits from \eqref{eq:conjClassesDecomp} a partition
\begin{equation}
\label{eq:RealConjClasses}
\pi_0(G \git_R \hat{G})
=
\pi_0(G \git G)_{-1} \sqcup \pi_0(G \git G)_{+1} \slash \Z/2.
\end{equation}

\bigskip

Let  $X$ be a $\hat{G}$-space.  Note that for each $g\in G$, the fixed point space $X^g$ is a $C^R_{\hat{G}}(g)$-space. In addition, the $\Lambda_G(g) $-action on $X^g$ as defined in \eqref{lambda_act_def} can extend to an action by $\Lambda^R_{\hat{G}}(g)$: \begin{equation}
    [r, \alpha]\cdot x:= \alpha\cdot x.
\end{equation} for any element $[r, \alpha] \in \Lambda^R_{\hat{G}}(g)$, any $x\in X^g$.

The Real loop groupoid $\hat{\Lambda} (X \git \hat{G}) $, as defined in \cite[Definition 2.6]{huanyoung2022}, adds the involution as morphisms into the groupoid $\Lambda(X\git G)$. And it is a double cover of the groupoid $\Lambda(X\git G)$.
In addition,
we have the Real version of the decomposition \eqref{lambda_decomp}, i.e. the decomposition of the groupoid $\hat{\Lambda} (X \git \hat{G})$ corresponding to the partition \eqref{eq:RealConjClasses}.
\begin{proposition}
\label{prop:loopGrpdModel}
There is an equivalence of $B\Z/2$-graded groupoids
\begin{equation}
\label{eq:conjClassDecompLambdaUnori}
\hat{\Lambda} (X \git \hat{G})
\cong
\coprod_{g \in \pi_0(G \git G)_{-1}} X^g \git \Lambda_{\hat{G}}^R(g)
\sqcup
\coprod_{g \in \pi_0(G \git G)_{+1} \slash \Z/2}
X^g \git \Lambda_{G}(g).
\end{equation}
\end{proposition}

The twisted Real quasi-elliptic cohomology is defined in \cite[Definition 3.2, Proposition 3.3]{huanyoung2022} in terms of  Freed-Moore K-theories.
\begin{definition}
\begin{equation}
\label{eq:QRDecomp}
\QR^{\ast+ \hat{\alpha}}(X \git G)
:= KR^{\bullet + \tilde{\tau}^{\refl}_{\pi}(\hat{\alpha})}(\Lambda ( X \git G )) \cong 
\prod_{g \in \pi_0(G \git_R \hat{G})} {^{\pi}}K^{\ast + \tilde{\tau}_{\pi}^{\refl}(\hat{\alpha})}_{\Lambda^R_{\hat{G}}(g)}(X^g), 
\end{equation} where $\hat{\alpha}$ is a fixed element in $ H^4(B\hat{G}; \Z)$ and $\tilde{\tau}_{\pi}^{\refl}$ is the Real transgression map.
\end{definition}

By the property of the Freed-Moore K-theory \cite{freed2013b}, if the Real structure $\widehat{G}$ splits, each factor in \eqref{eq:QRDecomp}  is 
the equivariant $KR$-theory defined by Atiyah and Segal \cite{atiyah1969}.

In addition, using the partition \eqref{eq:RealConjClasses}, the isomorphism \eqref{eq:QRDecomp} can be written as
\begin{equation}
\label{eq:QRDecompExpli}
\QR^{\ast + \hat{\alpha}}(X \git G)
\cong
\prod_{g \in \pi_0(G \git G)_{-1}}
KR^{\ast+\tilde{\tau}_{\pi}^{\refl}(\hat{\alpha})}_{\Lambda_{G}(g)}(X^g)
\times
\prod_{g \in \pi_0(G \git G)_{+1} \slash \Z/2}
K^{\ast+\tau(\alpha)}_{\Lambda_{G}(g)}(X^g).
\end{equation}

The $B \Z/2$-graded morphism $\hat{\Lambda} (X \git \hat{G}) \longrightarrow B  \mathrm{O}(2)$ which tracks loop rotation and reflection makes $\QR^{\ast}(X \git G)$ into a $KR^{\ast}_{\mathbb{T}}(\pt)$-algebra and, in particular, a module over $\Z[q^{\pm }] \subset KR^{\ast}_{\mathbb{T}}(\pt)$. 

\begin{theorem}
\label{thm:TateVsQR}
Assume that $\hat{G}$ is non-trivially $\Z/2$-graded. The relation between twisted Real quasi-elliptic cohomology 
 and twisted Real equivariant Tate K-theory is 
\[
KR_{Tate}^{\ast + \hat{\alpha}}(X \git G)
\cong
\QR^{\ast + \hat{\alpha} }(X \git G) \otimes_{KR^{\ast}(\pt)[q^{\pm }]} KR^{\ast}(\pt)((q)).
\]
\end{theorem}

In addition, we give an example computing Real quasi-elliptic cohomology, which is \cite[Example 3.7]{huanyoung2022}. The conclusions in Example \ref{QEllRcyclic} are applied in the computation of Section \ref{Real_spin3_s4} and Section \ref{prod_rc_qell}.
\begin{example} \label{QEllRcyclic}
Let $G = \Z/n = \langle r \rangle$ and $\hat{G} = D_{2n}$. The $\Z/2$-action on $\pi_0(\Z/n \git \Z/n) =\Z/n$ is trivial. By the isomorphism
\eqref{eq:QRDecompExpli}, 
\begin{equation}
\label{eq:RealZnDecompDihed}
\QR^{\ast}(\pt\git \Z/n)
\cong 
\prod_{m=0}^{n-1}
KR^{\ast}_{\Lambda_{\Z/n}(r^m)}(\pt).
\end{equation}
As discussed in \cite[Example 3.7]{huanyoung2022}, 
\begin{equation}
\label{eq:QRZnPt}
KR^{\ast}_{\Lambda_{\Z/ n}(r^m)}(\pt)
\cong 
KR^{\ast}(\pt)[q^{\pm}, x_m] / \langle x_m^n - q^m \rangle.
\end{equation}
\end{example}

\section{Real Quasi-elliptic cohomology of $S^4$ acted by a finite subgroup of $\mathrm{Spin}(3)$} \label{Real_spin3_s4}

In this section, we compute all the Real quasi-elliptic cohomology theories \[\QR^{\ast}_G(S^4)\] where $G$ goes over all the finite subgroups of $SU(2)\cong \mathrm{Spin}(3)$.

First we explain how the group $G$ acts on $S^4$.
We have the standard orthogonal $\mathrm{SO}(5)$-action on $\R^5$ and also on the subspace $S^4\subset \R^5$. The covering map \[\mathrm{Spin}(5) \longrightarrow \mathrm{SO}(5)\] makes $S^4$ a well-defined $\mathrm{Spin}(5)$-space. 
The $G$-action on $S^4$ is induced by  the composition 
\begin{equation} i_G: G\hookrightarrow \mathrm{Spin}(3)\buildrel{p_1}\over\longrightarrow \mathrm{Spin}(3)\times \mathrm{Spin}(3) =\mathrm{Spin}(4) \hookrightarrow \mathrm{Spin}(5) \label{gpact:def}\end{equation}  where $p_1$ is the projection to the first factor of the product group. 

We give the explicit formula of the $G$-action below.
The group $S(\mathbb{H})$ of unit quaternions is isomorphic to $SU(2)\cong \mathrm{Spin}(3)$ via the correspondence \[a+bi+cj+dk \mapsto \left[ {\begin{array}{cc} a+bi & c+di \\  -c+di & a-bi
\end{array}} \right].  \] In view of this, 
$\mathrm{Spin}(4)$ can be described as the group \[ \{ \left[ {\begin{array}{cc} q & 0 \\  0 & r
\end{array}} \right] \mid q, r\in \mathbb{H}, |q|=|r|=1.\}, \] and $\mathrm{Spin}(5)$ can be identified with the  quaternionic unitary group.
Thus, as indicated in \cite[pp.263]{porteous_1995}, the inclusion from $\mathrm{Spin}(4)\hookrightarrow \mathrm{Spin}(5)$ is given by the formula
\begin{equation} \left[ {\begin{array}{cc} q & 0 \\  0 & r
\end{array}} \right] \mapsto \left[ {\begin{array}{cc} q & 0 \\  0 & r
\end{array}} \right]. \label{spin4act} \end{equation}

In addition, as shown in \cite[pp.151]{porteous_1995}, the rotation of $\mathbb{R}^4$ represented by 
\[\left[ {\begin{array}{cc} q & 0 \\  0 & r
\end{array}} \right] \in \mathrm{Spin}(4)\] is given by the map \begin{equation} \label{spin(4)action} \left[ {\begin{array}{cc} y & 0 \\  0 & \overline{y}
\end{array}} \right] \mapsto 
\left[ {\begin{array}{cc} q & 0 \\  0 & r
\end{array}} \right] \left[ {\begin{array}{cc} y & 0 \\  0 & \overline{y}
\end{array}} \right] \widehat{\left[ {\begin{array}{cc} q & 0 \\  0 &r
\end{array}} \right]}^{-1} = \left[ {\begin{array}{cc} qy\overline{r} & 0 \\  0 & r\overline{y}\overline{q}
\end{array}} \right].\end{equation} where $\R^4$ is identified with the linear space \[\{ \left[ {\begin{array}{cc} y & 0 \\  0 & \overline{y}
\end{array}} \right] \mid  y\in \mathbb{H}.\} .\]

Then, the group $\mathrm{Spin}(4)\subset \mathrm{Spin}(5)$ acts on $S^4\subset \mathbb{R}^5$ via the composition \begin{equation} \mathrm{Spin}(4)\rightarrow \mathrm{SO}(4)\xrightarrow{A\mapsto \left[ {\begin{array}{cc} A & 0 \\  0 & 1
\end{array}} \right]}\mathrm{SO}(5) \label{spins5} \end{equation}
and the standard orthogonal action.

\bigskip

In the rest part of the paper, we will use the symbol \[A_\theta\] to denote the matrix \[
\left[ {\begin{array}{cc} e^{\theta i} & 0 \\  0 &e^{-\theta i} 
\end{array}} \right] ,\] and the symbol \[B_{\theta}\] to denote the matrix \[\left[ {\begin{array}{cc} \cos \theta & -\sin \theta  \\ \sin \theta & \cos \theta
\end{array}} \right]\]


First we need to pick a Real structure $(\widehat{SU(2)}, \pi)$  on the group $SU(2)$ as well as on all its finite subgroups by equipping the group with a reflection $s$. The choice is definitely not unique. Next, we define the reflection action on $S^4$ and, thus, together with \eqref{spin(4)action}, we define the action on $S^4$ by $\widehat{SU(2)}$.

\begin{example} \label{Real_str_SU(2)}
    

Motivated by the Real structure \[1\rightarrow \Z/n\rightarrow D_{2n}\buildrel{\pi}\over\rightarrow \Z/2 \rightarrow 1\] of the cyclic group $\Z/n < SU(2)$, 
we want to pick a Real structure $(\widehat{SU(2)}, \pi)$ on $SU(2)$ making the diagrams below commute.

\begin{equation} \label{Real:comm}
    \xymatrix{ 1 \ar[r] &\Z/n \ar[r] \ar@{^{(}->}[d] & D_{2n} \ar[r]\ar@{^{(}->}[d] &\Z/2 \ar[r] \ar@{=}[d] &1   \\
1\ar[r] &\mathrm{SO}(2) \ar@{^{(}->}[d] \ar[r] & O(2) \ar[r]^{\det} \ar@{^{(}->}[d] &\Z/2 \ar[r] \ar@{=}[d]& 1 \\
1\ar[r] &SU(2) \ar[r] & \widehat{SU(2)} \ar[r]^{\pi} &\Z/2 \ar[r]& 1
}
\end{equation} where the horizontal sequences are all exact. In the left column, the generator $r$ of the rotation group $ \Z/n < D_{2n}$ is mapped to the rotation  $B_{ \frac{2\pi  }{n}}$ in $\mathrm{SO}(2)$. 
The lower left vertical map can be chosen to map the rotation $B_{\frac{2\pi}{n}}$ to $A_{\frac{2\pi}{n}} \in SU(2)$. In addition, the reflection in $D_{2n}$ can be mapped to \[s:= \left[ {\begin{array}{cc} 0 & 1 \\ 1 & 0
\end{array}} \right] \in U(2).\] It's straightforward to check that $(sA_{\theta})^2$ is identity for any $\theta$.
In addition, we can take the action of $s$ on $\mathbb{R}^4\cong \mathbb{H}$ to be \begin{equation}\label{ref:su2:1}(a+bi+cj+dk) \mapsto (a-bi+cj-dk). \end{equation} Note that under the reflection \eqref{ref:su2:1}, the north and south poles of $S^4$ are still fixed. On the $\mathbb{R}^4$-plane,  the two pairs of points \[(0, 1, 0,0) \mbox{  and } (0, -1, 0, 0)\] \[ (0, 0, 0, 1) \mbox{   and  } (0, 0, 0, -1)\] are switched by the reflection respectively.  
It is straightforward to check that   $ (sA_{\theta})^2$ acts as identity on $\mathbb{R}^4$ for any $\theta$. 
Thus, it's reasonable to take the Real structure to be the subgroup \[ SU(2)\langle s \rangle \] of $U(2)$ and take the projection   to be the determinant
map $\det$.

\bigskip

Instead, we can map the rotation $r$ to the matrix $B_{\frac{2\pi}{n}}$, which is a conjugation of $A_{\frac{2\pi}{n}}$. We have \[ A^{-1} B_{\theta} A = A_{\theta}
\] where $A= \frac{1}{\sqrt{2}}\left[ {\begin{array}{cc} 1 & -i \\  -i & 1
\end{array}} \right] $ and $\theta$ is any real number. In addition, the reflection $s$ is fixed under the conjugation. The corresponding Real structure of $SU(2)$ is still $SU(2)\langle s \rangle $ and the diagram \eqref{Real:comm} still commutes.

\bigskip

Moreover, we'd like to mention a different  choice of the Real structure $\widehat{SU(2)}$. In the diagram \eqref{Real:comm}, we  map the rotation  $B_{ \frac{2\pi  }{n}}$  in $\mathrm{SO}(2)$ to the same matrix in $SU(2)$ but map the reflection   to \[s':= \left[ {\begin{array}{cc} 1 & 0 \\ 0 & -1
\end{array}} \right] \in U(2).\]  Note that $A^{-1}s' A \neq s'$, i.e. $s'$ is not a fixed point under the conjugation taking $B_{\theta}$ to $A_{\theta}$.  
We can check, for any $\theta$,  $(s' B_{\theta})^2 = I$.  The action of $s'$ on $\R^4$ can be  defined as \begin{equation}\label{ref:su2:2} (a+bi+cj+dk) \mapsto (a+bi-cj-dk). \end{equation}  Under the reflection \eqref{ref:su2:2}, the north and south poles are also fixed.  On $\R^4$,  the two pairs of points \[(0, 0, 1, 0) \mbox{  and } (0, 0, -1, 0)\] \[ (0, 0, 0, 1) \mbox{   and  } (0, 0, 0, -1)\] are switched by the reflection respectively.   It's straightforwards to check that $(s'B_{\theta})^2$ acts as identity on $\R^4$ for any $\theta$. 
Thus, it's reasonable to take the Real structure to be the subgroup \[ SU(2)\langle s' \rangle \] of $U(2)$ and the projection $\pi$  to be the determinant $\det$.

Since $SU(2)$ is a normal subgroup of $U(2)$,  both Real structures, $SU(2)\langle s\rangle$ and $SU(2)\langle s'\rangle$, split.


\end{example}

\begin{example} \label{Real:G:SU(2)}
    For any finite subgroup $G$ of $SU(2)$, \[ \hat{G}:= (G\langle s\rangle, \det) \]  is the restriction of the Real structure \[ (SU(2)\langle s\rangle, \det) \] of $SU(2)$ to $G$. It  defines a Real structure on $G$.

    Similarly, \[  \hat{G}':= (G\langle s'\rangle, \det) \] defines a Real structure on $G$. 
\end{example}

\begin{remark} We give in Example \ref{Real_str_SU(2)} some reasonable choices of reflection on the representation sphere $S^4$, which all keep the north pole and the south pole fixed.
We didn't find a canonical choice of reflection that switches the north pole and the south pole.

As indicated in \cite[p.215]{MarzantowiczPrieto04},   for $V$ a real vector space equipped with a linear $G$-action, stereographic projection exhibits a $G$-equivariant homeomorphism between the representation sphere $S^V := V_{\mathrm{cpt}}$ (the one-point compactification) and the unit sphere $S(V \oplus \mathbb{R}_{\mathrm{triv}})$ (where the $\mathbb{R}$-summand is equipped with the trivial $G$-action):
  \[
    S^V
    \,\simeq_{{}_G}\,
    S\big(
      V \oplus \mathbb{R}_{\mathrm{triv}}
    \big)
    \,.
  \]

A better choice of reflection on $S(V \oplus \mathbb{R}_{\mathrm{triv}})$ is that sending a point $(v, r)\in S(V \oplus \mathbb{R}_{\mathrm{triv}})$ to $(v, -r)$. The map corresponding to that on $S^V$, which is \[ \begin{cases} v \mapsto \frac{1}{\parallel v \parallel  } v , &\text{ if  }v\neq 0, \infty;
\\ \text{the north pole} \mapsto \text{ the south pole }, &\text{ if  }v= \infty; \\
\text{the south pole} \mapsto \text{ the north pole }, &\text{ if  }v = 0, \end{cases}\] wherer $\parallel  v \parallel $ is the lenghth of the vector. The map preserves angle  but not the length of the vector when it is not $1$, and, especially, it is not linear. 
    
\end{remark}

\begin{remark}
    We'd like to mention that the choice of the reflection in the Real structure is definitely not unique, neither is the choice of the action of it on $S^4$. Though different choices of the Real structure may lead to different $\QR^{\ast}_{G}(S^4)$, different choices of reflection action may lead to little difference. Indeed, in the computation of $\QR^{\ast}_{G}(S^4)$ with $G$ a finite subgroup of $SU(2)$, for most elements $g\in \pi_{0}(G\git G)$, the fixed point space $(S^4)^g$ consists only the north pole and the south pole, where the reflections, those in Example \ref{Real_str_SU(2)}, etc., act trivially. 
    
    In addition, for the identity element $e\in G$, $(S^4)^e = S^4$ is a representation sphere of the group $\Lambda_G(e)$. Thus, by  \cite[Theorem 5.1]{atiyah1968_Bott}, the computation of the corresponding factor $KR^{0}_{\Lambda_{G}(e)}((S^4)^e)$ can be reduced to that of the Real representation ring of $\Lambda_G(e)\cong G\times \mathbb{T}$.
    
\end{remark}

To compute the Real quasi-elliptic cohomology of 4-spheres \begin{equation}
\QR^{\ast}(S^4 \git G)
\cong
\prod_{g \in \pi_0(G \git G)_{-1}} KR^{\ast}_{\Lambda_{G}(g)}((S^4)^g) \times \prod_{g \in \pi_0(G \git G)_{+1}/\Z/2} K^{\ast}_{\Lambda_{G}(g)}((S^4)^g),
\end{equation} acted by a finite subgroup of \[G<SU(2)\cong \mathbb{H},\] we need to find all the fixed points in $ G $ under the involution, i.e. the Real conjugation. 
Below is a conclusion that makes the computation easier.


\begin{proposition} \label{fixedpt_invo}
    If we take the Real structure $\hat{G}'$ on a finite subgroup $G$ of $SU(2)$, for any element $\beta$ in $G$, we have the conclusions below.
    \begin{enumerate}
        \item $\beta$ is a fixed point under the involution $s'$ if and only if $s'\beta^{-1} s'$ is in the conjugacy class of $\beta$ in $G$.

        \item If there is an element in the conjugacy class of $\beta$ which is a unit quaternion and its coefficient of $i$ is zero, then we have $s'\beta^{-1}s' =\beta$ and $\beta$ is a fixed point under the involution. 
    \end{enumerate}
\end{proposition}

\begin{proof}
A given element $\beta \in G$ is a fixed point under the involution if and only if the set $C_G^R(\beta) \setminus C_G(\beta)$ is nonempty, i.e.  there is an element $x= s'y$ for some $y\in G$   satisfying 
\[x\beta x^{-1}= \beta^{-1}. \] So we get the first conclusion.  



Since $\beta$ is an element in $SU(2)$, thus, it has a quaternion representation $\beta = a +bi+cj+dk$. In (ii), we discuss a very special case that $s'\beta s'= \beta^{-1}$ exactly.
We start the computation below. 
\[ s' \beta s'^{-1}= 
    \left[ {\begin{array}{cc} 1 & 0 \\ 0 & -1
\end{array}} \right]  \left[ {\begin{array}{cc} a+bi & c+di \\ -c+di & a-bi
\end{array}} \right]  \left[ {\begin{array}{cc} 1 & 0 \\ 0 & -1
\end{array}} \right]  =  \left[ {\begin{array}{cc} a+bi & -c-di \\ c-di & a-bi
\end{array}} \right] 
\] The right hand side should be the inverse of $\beta$. So we establish the equation..
\[\left[ {\begin{array}{cc} a+bi & c+di \\ -c+di & a-bi
\end{array}} \right] \left[ {\begin{array}{cc} a+bi & -c-di \\ c-di & a-bi
\end{array}} \right] = \left[ {\begin{array}{cc} 1 & 0 \\ 0 & 1
\end{array}} \right]   \] Solving the equation, we get \[ \begin{cases}
    b&=0 \\
    a^2+c^2+d^2 &= 1
\end{cases}\] i.e. \[s'\beta s' = \beta^{-1}\] if and only if $\beta = a+bi+cj+dk$ is a unit quaternion with $b=0$.

\end{proof}

Similarly, we have the conclusion. 

\begin{proposition} \label{fixedpt_invo2}
    If we take the Real structure $\hat{G}$ on a finite subgroup $G$ of $SU(2)$, for any element $\beta$ in $G$, we have the conclusions below.
    \begin{enumerate}
        \item $\beta$ is a fixed point under the involution $s$ if and only if $s\beta^{-1} s$ is in the conjugacy class of $\beta$ in $G$.

        \item If there is an element in the conjugacy class of $\beta$ which is a unit quaternion and its coefficient of $k$ is zero, then we have $s\beta^{-1}s =\beta$ and $\beta$ is a fixed point under the involution.
    \end{enumerate}
\end{proposition} The proof is analogous to that of Proposition \ref{fixedpt_invo}.

\bigskip

Next we will compute $\QR^{\ast} (S^4\git G)$ with $G$ a finite subgroup of $SU(2)$ one by one. 
Before that  we recall the  classification of the finite subgroups of $\mathrm{Spin}(3)\cong SU(2)$. There are many references for the classification, \cite[Chapter XIII]{dickson2014algebraic}, \cite{Stekolshchik_Coxeter_Mckay}, \cite{nlab:finite_rotation_group} etc. The finite subgroups of $SU(2)$ are classified as: \begin{itemize}\item 
the cyclic group of order $n$\[G_n:= \{\left[ {\begin{array}{cc} \cos\frac{2\pi k}{n} & \sin\frac{2\pi k}{n} \\  -\sin\frac{2\pi k}{n} & \cos\frac{2\pi k}{n}
\end{array}} \right] \mid k\in\Z \};\]
\item  the dicyclic group of order $4n$ \[2D_{2n}:=\langle A_{\frac{2\pi}{2n}}, \left[ {\begin{array}{cc} 0 & 1 \\  -1 & 0
\end{array}} \right] \rangle; \]  
\item  the binary tetrahedral group $E_6$;
\item the binary octahedral group $E_7$; 
\item the binary icosahedral group $E_8$; \end{itemize}
where $n$
is any positive integer.

\begin{example} \label{G_nDihedralS^4} 

In this example we compute $\QR^{\ast} (S^4\git G_n)$ where $G_n$ is the finite cyclic subgroup
 \[\{\left[ {\begin{array}{cc} \cos\frac{2\pi k}{n} & \sin\frac{2\pi k}{n} \\  -\sin\frac{2\pi k}{n} & \cos\frac{2\pi k}{n}
\end{array}} \right] \mid k\in\Z \} < SU(2).\] 
We take the Real structure $\hat{G}'_n$ as defined in Example \ref{Real:G:SU(2)}, i.e. the group below together with the determinant map $\det$
\[ \langle G_n, \left[ {\begin{array}{cc} 1 & 0 \\  0 & -1 
\end{array}} \right] \rangle. \] It is isomorphic to the dihedral group $D_{2n}$. 
The involution on $\pi_0(G_n\git G_n)$ is trivial.

Thus, by \cite[Example 3.7]{huanyoung2022}, we get directly that \[\QR^{\ast}(S^4\git G_n) \cong \prod_{m=0}^{n-1} KR^{\ast}_{\Lambda_{G_n} (B_{ \frac{2\pi  m}{n}})} ((S^4)^{B_{ \frac{2\pi m }{n}}})
\cong KR^{\ast}_{G_n\times \mathbb{T}} (S^4)\oplus \prod_{m=1}^{n-1} KR^{\ast}_{\Lambda_{G_n} (B_{ \frac{2\pi m }{n}})} (S^0) \] where  $S^0$ consists of the fixed points, i.e. the south pole and the north pole of $S^4$.
Thus, \[ \prod_{m=1}^{n-1} KR^{\ast}_{\Lambda_{G_n} (B_{ \frac{2\pi m }{n}})} (S^0)  \cong  \prod_{m=1}^{n-1} KR^{\ast}_{\Lambda_{G_n} (B_{ \frac{2\pi  m}{n}})} (\pt) \oplus KR^{\ast}_{\Lambda_{G_n} (B_{ \frac{2\pi m }{n}})} (\pt)\] and by \cite[Example 3.7]{huanyoung2022}, the right hand side is isomorphic to  \[  \prod_{m=1}^{n-1} KR^{\ast}(\pt)[x, q^{\pm}] \slash \langle x^n-q^m \rangle \oplus  KR^{\ast}(\pt)[x, q^{\pm}] \slash \langle x^n-q^m \rangle. \]

In addition, by \cite[Theorem 5.1]{atiyah1968_Bott}, 
\begin{align*}  KR^{\ast}_{G_n\times \mathbb{T}} (S^4) & \cong KR^{\ast}_{G_n\times \mathbb{T}}(S^0)  \cong KR^{\ast}_{G_n\times \mathbb{T}}(\pt) \oplus
KR_{G_n\times \mathbb{T}}(\pt) \\ &\cong KR^{\ast}(\pt)[x, q^{\pm}] \slash \langle x^n-1  \rangle \oplus  KR^{\ast}(\pt)[x, q^{\pm}] \slash \langle x^n-1  \rangle. \end{align*}

In conclusion, \[\QR^{\ast}(S^4\git G_n) \cong   \prod_{m=0}^{n-1} KR^{\ast}(\pt)[x, q^{\pm}] \slash \langle x^n-q^m \rangle \oplus  KR^{\ast}(\pt)[x, q^{\pm}] \slash \langle x^n-q^m \rangle.\]
\end{example}

\begin{example} \label{2D2nS^4}

In this example we compute $\QR^{\ast} (S^4\git 2D_{2n})$ where  $2D_{2n}$ is the dicyclic group \[\langle A_{\frac{2\pi}{2n}}, \tau \rangle,\]
where $\tau$ is the reflection \[\left[ {\begin{array}{cc} 0 & -1 \\  1 & 0
\end{array}} \right] .\]
We take the Real structure $\hat{2D_{2n}}$ on $2D_{2n}$, as defined in Example \ref{Real:G:SU(2)}.  

In $2D_{2n}$  there are $n+3$ conjugacy classes. They are:
\begin{enumerate}
\item $\{I\}$, \item $\{-I\}$, \item 
$\{A_{\frac{\pi}{n}}, A^{-1}_{\frac{\pi}{n}}\}$, 
$\{A^2_{\frac{\pi}{n}}, A^{-2}_{\frac{\pi}{n}}\}$, $\cdots$, $\{A^{n-1}_{\frac{\pi}{n}}, A^{-(n-1)}_{\frac{\pi }{n}}\}$,  
\item 
$\{\tau, \tau A^2_{\frac{\pi}{n}}, \tau A^4_{\frac{\pi}{n}}\cdots \tau A^{2n-2}_{\frac{\pi}{n}}\} $, 
\item $\{\tau A_{\frac{\pi}{n}}, \tau A^3_{\frac{\pi}{n}}, \cdots \tau A^{2n-1}_{\frac{\pi}{n}} \}$, \end{enumerate} where the first two form the centre of the group. 

By Proposition \ref{fixedpt_invo2}, all the conjugacy classes  are fixed points under the reflection $s$.
Next we compute  below the factor in $\QR^{\ast}(S^4 \git {2D_{2n}})$ corresponding to each conjugacy class below.

\begin{enumerate}
    \item First we consider the Real conjugacy class represented by $I$.
The centralizer $C_{2D_{2n}}(I)= 2D_{2n}$ and the Real centralizer is the same \[C^R_{2\hat{D}_{2n}}(I)=2\hat{D}_{2n}.\] The group $\Lambda^R_{2\hat{D}_{2n}}(I) = \mathbb{R}\rtimes_{\pi} 2\hat{D}_{2n} /\langle(-1, I) \rangle.$ 
By \cite[Theorem 5.1]{atiyah1968_Bott}, 
\begin{align*}
    KR^{\ast}_{\Lambda_{2D_{2n}}(I)}((S^4)^I) &\cong KR^{\ast}_{\mathbb{T}\times  2D_{2n}}(S^4)
    \cong KR^{\ast}_{\mathbb{T}\times  2D_{2n}}(S^0) \\
    &\cong KR^{\ast}_{\mathbb{T}\times 2D_{2n}}(\pt)\oplus KR^{\ast}_{\mathbb{T}\times 2D_{2n}}(\pt) \\
    &\cong KR^{\ast}_{2D_{2n}}(\pt)[q^{\pm}]\oplus KR^{\ast}_{2D_{2n}}(\pt)[q^{\pm}].
\end{align*} Note that $\Lambda^R_{2\hat{D}_{2n}}(I)$ is a Real structure on $\mathbb{T}\times 2D_{2n}$.

  \item  Then we consider the Real conjugacy class represented by $-I$. 
In this case, the centralizer $C_{2D_{2n}}(- I)= 2D_{2n}$ and the Real centralizer \[C^R_{2\hat{D}_{2n}}(- I)=2\hat{D}_{2n}.\]
We have the Real central extension \[ 1\longrightarrow \Z/2 \longrightarrow \Lambda^R_{2\hat{D}_{2n}}(-I) \longrightarrow \Lambda^R_{\hat{D}_{2n}}(I)\longrightarrow 1\] 

By Corollary   \ref{RM:Z2},
    \begin{align*}
    KR^{\ast}_{\Lambda_{2D_{2n}}(-I)}((S^4)^{-I}) &\cong KR^{\ast}_{\Lambda_{2D_{2n}}(-I)} (S^0) \\
     &\cong \prod^2_1 KR^{\ast}_{\Lambda_{2D_{2n}}(-I)} (\pt)\\ 
    &\cong \prod_{1}^{2} KR^{\ast }_{\Lambda_{D_{2n}}(I)}(\pt)\oplus KR^{\ast + \hat{\nu}_{\Lambda^R_{2\hat{D}_{2n}}(-I), sign}}_{\Lambda_{D_{2n}}(I)}(\pt) \\
    & \cong \prod_{1}^{2} KR^{\ast }_{D_{2n}}(\pt)[q^{\pm}]\oplus KR^{\ast + \hat{\nu}_{\Lambda^R_{2\hat{D}_{2n}}(-I), sign}}_{D_{2n}}(\pt)[q^{\pm}],
\end{align*} where $sign$ is the sign representation of $\Z/2$.

\item Then we compute the factor in $\QR^{\ast}(S^4 \git {2D_{2n}})$ corresponding to $A_{\frac{2\pi m}{2n}}$ which is not $\pm I$.

The centralizer $C_{2D_{2n}}(A_{\frac{2\pi m}{2n}})$ is the cyclic group $\langle A_{\frac{2\pi}{2n}}\rangle \cong \Z/(2n)$. The Real centralizer \[C^R_{2\hat{D}_{2n}} (A_{\frac{2\pi m }{2n}} ) =D_{4n}\] is the dihedral group of order $4n$. 
In this case, by \cite[Example 3.7]{huanyoung2022}, \begin{align*}
KR^{\ast}_{\Lambda_{2D_{2n}}(A_{\frac{2\pi m}{2n}} )}(S^4)^{A_{\frac{2\pi m}{2n}} } &\cong KR^{\ast}_{\Lambda_{2D_{2n}}(A_{\frac{2\pi m}{2n}} )}(S^0) \cong KR^{\ast}_{\Lambda_{2D_{2n}}(A_{\frac{2\pi m}{2n}} )}(\pt) \oplus KR^{\ast}_{\Lambda_{2D_{2n}}(A_{\frac{2\pi m }{2n}} )}(\pt) \\
& \cong KR^{\ast}(\pt)[x, q^{\pm}] /\langle x^{2n} - q^{2m}\rangle \oplus KR^{\ast}(\pt)[x, q^{\pm}] /\langle x^{2n} - q^{2m}\rangle.
\end{align*}

\item Then we compute the factor corresponding to the conjugacy class represented by $\tau$. The centralizer  $C_{2D_{2n}}(\tau)= \langle \tau \rangle \cong \Z/4$ and the Real centralizer \[C^R_{2\hat{D}_{2n}}(\tau) = \langle \tau, s  \rangle \cong D_4. \] Thus, \begin{align*}
    KR^{\ast}_{\Lambda_{2D_{2n}}(\tau)} (S^4)^{\tau}&\cong KR^{\ast}_{\Lambda_{\Z/4} (1)}(S^0)  \\
    &\cong RR\Lambda_{\Z/4}(1) \oplus RR\Lambda_{\Z/4}(1) \\
    &\cong KR^{\ast}(\pt)[x, q^{\pm}]/\langle x^4- q\rangle \oplus KR^{\ast}(\pt)[x, q^{\pm}]/\langle x^4- q\rangle .
\end{align*}

    \item 
    For the conjugacy class represented by $\tau A_{\frac{\pi}{n}}$, the centralizer $C_{2D_{2n}}(\tau A_{\frac{\pi}{n}})= \langle \tau A_{\frac{\pi}{n}} \rangle 
\cong \Z\slash 4$ and the Real centralizer  $C^R_{2\hat{D}_{2n}}(\tau A_{\frac{\pi}{n}})= \langle \tau A_{\frac{\pi}{n}}, s\tau \rangle 
\cong D_4$. 
Then, the factor corresponding to $\tau A_{\frac{\pi}{n}} $ is \begin{align*}
    KR^{\ast}_{\Lambda_{2D_{2n}}(\tau A_{\frac{\pi}{n}})} (S^4)^{\tau A_{\frac{\pi}{n}}}&\cong KR^{\ast}_{\Lambda_{\Z/4} (1)}(S^0)  \\
    &\cong RR\Lambda_{\Z/4}(1) \oplus RR\Lambda_{\Z/4}(1) \\
    &\cong KR^{\ast}(\pt)[x, q^{\pm}]/\langle x^4- q\rangle \oplus KR^{\ast}(\pt) [x, q^{\pm}]/\langle x^4- q\rangle .
\end{align*}
\end{enumerate}
Thus, in conclusion,
\begin{align*}
    \QR^{\ast} (S^4 \git 2D_{2n}) = &KR^{\ast}_{\Lambda_{2D_{2n}}(I)}((S^4)^I) \times KR^{\ast}_{\Lambda_{2D_{2n}}(-I)}((S^4)^{-I})  \\
    &\times \prod_{m=1}^{n-1}KR^{\ast}_{\Lambda_{2D_{2n}}(A^m_{\frac{\pi }{n}})} ((S^4)^{A^m_{\frac{\pi }{n}}}) \\
    &\times KR^{\ast}_{\Lambda_{2D_{2n}}(\tau)}((S^4)^{\tau}) 
    \times  KR^{\ast}_{\Lambda_{2D_{2n}}(\tau  A_{\frac{\pi}{n}})}((S^4)^{\tau  A_{\frac{2\pi}{2n}}}) \\
    \cong  &KR^{\ast}_{2D_{2n}}(\pt)[q^{\pm}]\oplus KR^{\ast}_{2D_{2n}}(\pt)[q^{\pm}] \\ 
    & \times \prod_{1}^{2} KR^{\ast }_{D_{2n}}(\pt)[q^{\pm}]\oplus KR^{\ast + \hat{\nu}_{\Lambda^R_{2\hat{D}_{2n}}(-I), sign}}_{D_{2n}}(\pt)[q^{\pm}]\\ 
    &\times \prod_{m=1}^{n-1}  KR^{\ast}(\pt)[x, q^{\pm}]/\langle x^{2n}- q^{2m}\rangle \oplus KR^{\ast}(\pt)[x, q^{\pm}]/\langle x^{2n}- q^{2m}\rangle \\ &\times  KR^{\ast}(\pt)[x, q^{\pm}]/\langle x^4-q\rangle \oplus KR^{\ast}(\pt)[x, q^{\pm}]/ \langle x^4-q\rangle \\
    & \times KR^{\ast}(\pt)[x, q^{\pm}]/\langle x^4-q\rangle \oplus KR^{\ast}(\pt)[x, q^{\pm}]/\langle x^4-q\rangle,
\end{align*}  where $sign$ is the sign representation of $\Z/2$.


\end{example}

\begin{example} \label{E6_S4_Real}

In this example we compute $ \QR^{\ast}(S^4\git E_6) $ where $E_6$
is  the binary tetrahedral group $E_6$.
We take the Real structure $\hat{E_6'}$ on it, i.e. \[\hat{E_6'}=  E_6\langle s' \rangle.\] 
The quaternion representation of $E_6$ is given explicitly at \cite{Phillips_Tetrahedral} and  \cite{QR:BiTetraGrp}.

We can compute the conjugacy classes in $E_6$ explicitly. A multiplication table for the binary tetrahedral group is given here \cite{multi:BiTetraGrp}. For the convenience of the readers, we apply the same symbols of the elements as those in \cite{multi:BiTetraGrp} and \cite{QR:BiTetraGrp}.  A list of representatives are given in Figure \ref{E6:conj}. This list can be obtained by direct computation. In addition, by Proposition \ref{fixedpt_invo}, an element in $E_6$  represents a fixed point in $\pi_0(E_6\git_R \hat{E_6})$ if and only if it is $\pm I$, $\pm i$, $\pm j$ or $\pm k$. 
Note that,  for $E_6$, if we take the Real structure $\hat{E_6}$, we will get the same set of fixed points under the reflection.

\begin{figure} \begin{center} 
\begin{tabular}{|c | c | c | c|} 
 \hline 
 A representative of   & Conjugacy class & Order & Fixed point under \\
 the conjugacy class  &&& the involution?
 \\ \hline
$1$ & $ \{ 1 \}$ & $1$ & Y \\
$-1$ & $\{-1\}$ & $2$ & Y \\
$j$ & $\{\pm i, \pm j,  \pm k\} $ & $4$ & Y \\
$a$ & $\{a, b, c, d\} $ & $6$  & N \\
$-a$ & $\{-a, -b, -c, -d\}$ & $3$  & N \\
$a^2$ & $\{a^2, b^2, c^2, d^2\}$ & $3$  & N \\
$-a^2$ & $ \{-a^2, -b^2, -c^2, -d^2\} $  & $6$  & N \\
 \hline
\end{tabular} \caption{Conjugacy classes of $E_6$}\label{E6:conj}
\end{center} \end{figure}

Below we compute the factors of $\QR_{E_6}(S^4)$ corresponding to each conjugacy class respectively.

\begin{enumerate}
\item For the conjugacy class represented by $I$, the Real centralizer $C^R_{\hat{E_6'}}(I) = \hat{E_6'}$. By \cite[Theorem 5.1]{atiyah1968_Bott}, we have
\begin{align*}
KR^{\ast }_{\Lambda_{E_6}(I)} ((S^4)^I) &\cong KR^{\ast}_{E_6\times \mathbb{T}}(S^4) \cong KR^{\ast}_{E_6\times \mathbb{T}}(S^0) \\
&\cong KR^{\ast}_{E_6\times \mathbb{T}}(\pt) \oplus KR^{\ast}_{E_6\times \mathbb{T}}(\pt)  \\
&\cong KR^{\ast}_{E_6}(\pt) [q^{\pm}] \oplus KR^{\ast}_{E_6}(\pt) [q^{\pm}]
\end{align*}

\item For the conjugacy class represented by $-I$, we have $(S^4)^{-I} = S^0$. 

Let $\hat{T'_6}$ denote the group $T_6\langle  s' \rangle$. We have the short exact sequence \[1\rightarrow \Z/2 \rightarrow \hat{T_6'} \rightarrow T_6\rightarrow 1\]
Especially, we have the commutative diagram below:
\begin{equation}
\xymatrix{ 0\ar[r] &\mathbb{Z}/2 \ar[r] \ar@{=}[d] &E_6 \ar[r]^{\pi} \ar@{^{(}->}[d] &T_6\ar[r] \ar@{^{(}->}[d] &0 \\ 
0\ar[r] &\mathbb{Z}/2 \ar[r] &\hat{E'_6} \ar[r]^{\pi} & \hat{T_6'}\ar[r] &0 }
\end{equation}

Note that we have the short exact sequence
\[0\rightarrow \Z/2 \longrightarrow \Lambda^R_{\hat{E_6'}}(-I) \buildrel{[(\pi, id), id ]}\over\longrightarrow  \Lambda^R_{\hat{T_6'}}(I)\longrightarrow 0.  \] By   Corollary \ref{RM:Z2}, 
\begin{align*}
KR^{\ast}_{\Lambda_{E_6}(-I)} ((S^4)^{-I}) &\cong KR^{\ast}_{\Lambda_{E_6}(-I)} (S^0) \\
&\cong   \prod_1^2KR^{\ast}_{\Lambda_{E_6}(-I)}(\pt)\\
 &\cong
\prod_{1}^{2} KR^{\ast}_{T_{6}\times \mathbb{T}}(\pt)\oplus KR^{\ast + \hat{\nu}_{\Lambda^R_{\hat{E_6'}(-I)}, sign}}_{T_{6}\times \mathbb{T}}(\pt) \\
&\cong \prod_{1}^{2} KR^{\ast}_{T_{6}}(\pt)[q^{\pm}]\oplus KR^{\ast + \hat{\nu}_{\Lambda^R_{\hat{E_6'}(-I)}, sign}}_{T_{6}}(\pt) [q^{\pm}], 
\end{align*} where $sign$ is the sign representation of $\Z/2$.

\item For the conjugacy class represented by $j$, $(S^4)^j=S^0$.
The centralizer $C_{E_6}(j) = \langle j \rangle \cong \Z/4$ and the Real centralizer 
\[ C^R_{\hat{E'_6}}(j) =C_{E_6}(j)\langle s' \rangle \cong D_4.\] Thus,
\begin{align*}
KR^{\ast}_{\Lambda_{E_6}(j)}((S^4)^j) &\cong KR^{\ast}_{\Lambda_{\mathbb{Z}/4}(1)}(S^0) \cong KR^{\ast}_{\Lambda_{\mathbb{Z}/4}(1)}(\pt) \oplus KR^{\ast}_{\Lambda_{\mathbb{Z}/4}(1)}(\pt) 
\\ &\cong KR^{\ast}(\pt)[x, q^{\pm}]/\langle x^4- q\rangle \oplus KR^{\ast}(\pt)[x, q^{\pm}]/\langle x^4- q\rangle.
\end{align*}

\item For the conjugacy class represented by $a$, we have 
\begin{align*}
K_{\Lambda_{E_6}(a)}((S^4)^a) &\cong K_{\Lambda_{\mathbb{Z}/6}(1)}(S^0) \cong R(\Lambda_{\mathbb{Z}/6}(1)) \oplus R(\Lambda_{\mathbb{Z}/6}(1)) 
\\ &\cong \Z[x, q^{\pm}]/\langle x^6- q\rangle \oplus \Z[x, q^{\pm}]/\langle x^6- q\rangle.
\end{align*}

\item For the conjugacy class represented by $-a$, we have \begin{align*}
K_{\Lambda_{E_6}(-a)} ((S^4)^{-a}) &\cong K_{\Lambda_{\Z/6}(4)} (S^0) \cong R(\Lambda_{\Z/6}(4))\oplus R (\Lambda_{\Z/6}(4)) \\
&\cong \Z[x, q^{\pm}] /\langle x^6- q^4\rangle \oplus \Z[x, q^{\pm}] /\langle x^6- q^4\rangle.    
\end{align*}

\item For the conjugacy class represented by $a^2$, we have \begin{align*}
K_{\Lambda_{E_6}(a^2)} ((S^4)^{a^2}) &\cong K_{\Lambda_{\Z/6}(2)} (S^0) \cong R(\Lambda_{\Z/6}(2))\oplus R (\Lambda_{\Z/6}(2)) \\
&\cong \Z[x, q^{\pm}] /\langle x^6- q^2\rangle \oplus \Z[x, q^{\pm}] /\langle x^6- q^2\rangle    
\end{align*}

\item For the conjugacy class represented by $-a^2$, we have 
\begin{align*}
K_{\Lambda_{E_6}(-a^2)}((S^4)^{-a^2}) &\cong K_{\Lambda_{\mathbb{Z}/6}(5)}(S^0) \cong  R(\Lambda_{\mathbb{Z}/6}(5)) \oplus R(\Lambda_{\mathbb{Z}/6}(5))
\\ &\cong \Z[x, q^{\pm}]/\langle x^6- q^5\rangle \oplus \Z[x, q^{\pm}]/\langle x^6- q^5\rangle .  
\end{align*}

\end{enumerate}

Thus, in conclusion, \begin{align*}
    \QR^{\ast}(S^4 \git E_6) = &KR^{\ast}_{\Lambda_{E_6}(1)} ((S^4)^1) \times KR^{\ast}_{\Lambda_{E_6}(-1)} ((S^4)^{-1}) \times KR^{\ast}_{\Lambda_{E_6}(j)}((S^4)^j) \\
    & \times K^{\ast}_{\Lambda_{E_6}(a)}((S^4)^a) \times K^{\ast}_{\Lambda_{E_6}(-a)} ((S^4)^{-a}) 
       \times K^{\ast}_{\Lambda_{E_6}(a^2)} ((S^4)^{a^2}) \\ &\times K^{\ast}_{\Lambda_{E_6}(-a^2)}((S^4)^{-a^2})
\\
\cong 
   & KR^{\ast}_{E_6}(\pt) [q^{\pm}] \oplus KR^{\ast}_{E_6}(\pt) [q^{\pm}] \\
  &\times \prod_{1}^{2} KR^{\ast}_{T_{6}}(\pt)[q^{\pm}]\oplus KR^{\ast + \hat{\nu}_{\Lambda^R_{\hat{E_6'}(-I)}, sign}}_{T_{6}}(\pt) [q^{\pm}]
   \\ & \times   KR^{\ast}(\pt)[x, q^{\pm}]/\langle x^4- q\rangle \oplus  KR^{\ast}(\pt)[x, q^{\pm}]/\langle x^4- q\rangle \\
   & \times K^{\ast}(\pt)[x, q^{\pm}]/\langle x^6- q\rangle \oplus K^{\ast}(\pt)[x, q^{\pm}]/\langle x^6- q\rangle \\
    &\times  K^{\ast}(\pt)[x, q^{\pm}] /\langle x^6- q^4\rangle \oplus K^{\ast}(\pt)[x, q^{\pm}] /\langle x^6- q^4\rangle \\
    &\times K^{\ast}(\pt)[x, q^{\pm}] /\langle x^6- q^2\rangle \oplus K^{\ast}(\pt)[x, q^{\pm}] /\langle x^6- q^2\rangle\\
&    \times K^{\ast}(\pt)[x, q^{\pm}]/\langle x^6- q^5\rangle \oplus K^{\ast}(\pt)[x, q^{\pm}]/\langle x^6- q^5\rangle.
\end{align*} where $sign$ is the sign representation of $\Z/2$.

\end{example}

\begin{example} \label{E7_S4_Real}
In this example we compute $\QR^{\ast}(S^4 \git E_7)$ where $E_7$ is the  binary octahedral
group. We take the Real structure $\hat{E_7'}$ on it, i.e. $E_7\langle s'\rangle$.

A presentation of $E_7$ is given as  
\[E_7= \langle \theta, t \mid r^2=\theta^3= t^4 = r\theta t = -1\rangle.  \]
We can get immediately that $r=\theta t$. Equivalently, there is a quaternion presentation of $E_7$ given by the embedding  \[E_7 \rightarrow  \mathbb{H}\] 
 sending $\theta$ to
$\frac{1}{2}(1+i+j+k)$,  $t$ to $\frac{1}{\sqrt{2}}(1+i)$, and $r$ to $\frac{1}{\sqrt{2}} (i+j)$. 

By \cite{McKay1980} and direct computation, 
we get Figure \ref{E7:conj:c:fps}, which provides a list of the representatives of the conjugacy classes of $E_7$, the centralizers of each representative, and the corresponding fixed point spaces.

\begin{figure}
\begin{center}
\begin{tabular}{|c | c | c | c| c|} 
 \hline Representatives $\beta$  &Centralizers  & Conjugacy class & Fixed points under & $(S^4)^{\beta}$  \\
 of Conjugacy classes & $C_{E_7}(\beta)$ & & the involution? & 
 \\ \hline
$1$ & $ E_7$  & $\{ 1 \} $ &Y & $S^4$ \\
$-1$ & $E_7$ & \{ -1\} &Y & $ S^0$  \\
$j= \theta t^2\theta^{-1}$ & $\langle \theta t \theta^{-1} \rangle \cong \mathbb{Z}/8$ & $\{\pm i, \pm j, \pm k\}$ & Y  & $S^0$  \\
$\theta$ & $\langle \theta \rangle \cong \mathbb{Z}/6$ & $\{ \frac{(1 \pm i \pm j \pm k)}{2} \}$ &Y & $S^0$\\
$-\theta  = \theta^4$ & $\langle \theta \rangle \cong \mathbb{Z}/6$ & $\{ \frac{(-1 \pm i \pm j \pm k)}{2} \} $ & Y & $S^0$ \\
$r$ & $\langle r\rangle \cong \mathbb{Z}/4$ &$ \{ \frac{1}{\sqrt{2}} (\pm i \pm j), \frac{1}{\sqrt{2}} (\pm i\pm k ),$ &Y & $S^0$ \\
&&$ \frac{1}{\sqrt{2}} (\pm j \pm k)\}$ &&\\
$t$ & $\langle t\rangle \cong \mathbb{Z}/8$ & $\{\frac{ 1\pm i}{\sqrt{2}},  \frac{ 1\pm j}{\sqrt{2}}, \frac{  1\pm k}{\sqrt{2}}\}$ &Y & $S^0$ \\
$-t = t^5$ & $\langle t\rangle \cong \mathbb{Z}/8$ & $ \{\frac{ - 1\pm i}{\sqrt{2}},  \frac{ - 1\pm j}{\sqrt{2}}, \frac{ -  1\pm k}{\sqrt{2}}\}$& Y  & $S^0$ \\
 \hline
\end{tabular} \caption{Conjugacy classes, centralizers and fixed point spaces} \label{E7:conj:c:fps}
\end{center} \end{figure}

Below we give the factor of $\QR^{\ast}(S^4 \git E_7)$ corresponding to each conjugacy class.
\begin{enumerate}
    \item For the conjugacy class represented by $I$, the Real centralizer $C^R_{\hat{E_7'}}(I) = \hat{E_7'}$. The factor corresponding to $I$
    \begin{align*}KR^{\ast}_{\Lambda_{E_7}(I)}((S^4)^I) &\cong KR^{\ast}_{E_7\times \mathbb{T}} (S^4) \cong KR^{\ast}_{E_7\times \mathbb{T}} (S^0)  \\
    &\cong KR^{\ast}_{E_7\times \mathbb{T}}(\pt) \oplus KR^{\ast}_{E_7\times \mathbb{T}}(\pt) \cong KR^{\ast}_{E_7}(\pt)[q^{\pm}]
    \oplus  KR^{\ast}_{E_7}(\pt)[q^{\pm}].\end{align*}

\item For the conjugacy class represented by $-I$, the Real centralizer $C^R_{\hat{E_7'}}(-I)= \hat{E_7'}$. 
Let  $T_7$ denote the chiral octahedral group and $\hat{T_7'}$ the Real structure $T_7\langle s'\rangle $.
And we have the commutative diagram 
\begin{equation}
\xymatrix{ 0\ar[r] &\mathbb{Z}/2 \ar[r] \ar@{=}[d] &E_7 \ar[r]^{\pi} \ar@{^{(}->}[d] &T_7\ar[r] \ar@{^{(}->}[d] &0 \\ 
0\ar[r] &\mathbb{Z}/2 \ar[r] &\hat{E_7'} \ar[r]^{(\pi, id)} &\hat{T_7'}\ar[r] &0 }
\end{equation}

Thus, by Corollary \ref{RM:Z2}, 
\begin{align*}
    KR^{\ast}_{\Lambda_{E_7}(-I)}(S^4)^{-I} &\cong KR^{\ast}_{\Lambda_{E_7}(-I)}(S^0) \\
    &\cong \prod_1^2 KR^{\ast}_{\Lambda_{E_7}(-I)}(\pt)
\\ &\cong 
    \prod_{1}^{2}KR^{\ast}_{T_{7}\times \mathbb{T}}(\pt) \oplus KR^{\ast + \hat{\nu}_{\Lambda^R_{\hat{E_7'}(-I)}, sign}}_{T_{7}\times \mathbb{T}}(\pt) 
    \\ &\cong     \prod_{1}^{2}KR^{\ast}_{T_{7}}(\pt) [q^{\pm}] \oplus KR^{\ast + \hat{\nu}_{\Lambda^R_{\hat{E_7'}(-I)}, sign}}_{T_{7}}(\pt) [q^{\pm}]
\end{align*} where $sign$ is the sign representation of $\Z/2$.

    \item 

    For the conjugacy class represented by $j$ is $\{\pm i, \pm j, \pm k\}$, its Real centralizer \[C_{\hat{E'_7}}^R(i) \cong D_8.\]
    Thus, $KR^{\ast}_{\Lambda_{E_7}(i)} ((S^4)^i) $ is isomorphic to \begin{align*}
        & KR^{\ast}_{\Lambda_{\mathbb{Z}/8}(2) }(S^0) \cong KR^{\ast}_{\Lambda_{\mathbb{Z}/8}(2) }(\pt)  \oplus KR^{\ast}_{\Lambda_{\mathbb{Z}/8}(2) }(\pt) \\
      \cong   & KR^{\ast}(\pt)[x, q^{\pm }] /\langle x^8-q^2\rangle \oplus KR^{\ast}(\pt)[x, q^{\pm }] /\langle x^8-q^2\rangle.
    \end{align*}

    \item For the conjugacy class represented by $\theta= \frac{1}{2}(1+i+j+k) $, the Real centralizer \[C_{\hat{E'_7}}^R(\theta) = \langle \theta, \frac{j+k}{\sqrt{2}} s' \rangle \cong D_6.\] Note that $(\frac{j+k}{\sqrt{2}} s')^2 =1 $ and $(\frac{j+k}{\sqrt{2}} s' \theta)^2 =1$. 
    Then $KR^{\ast}_{\Lambda_{E_7}(\theta)} ((S^4)^{\theta}) $ is isomorphic to \begin{align*}
        & KR^{\ast}_{\Lambda_{\mathbb{Z}/6}(1) }(S^0) \cong KR^{\ast}_{\Lambda_{\mathbb{Z}/6}(1) }(\pt ) \oplus KR^{\ast}_{\Lambda_{\mathbb{Z}/6}(1) }(\pt)\\
      \cong   & KR^{\ast}(\pt)[x, q^{\pm }] /\langle x^6-q\rangle \oplus KR^{\ast}(\pt)[x, q^{\pm }] /\langle x^6-q\rangle.
    \end{align*}

\item For the conjugacy class represented  $-\theta= -\frac{1}{2}(1+i+j+k) $,  the Real centralizer \[C_{\hat{E_7'}}^R(-\theta) = \langle -\theta, \frac{j+k}{\sqrt{2}} s' \rangle\cong D_6.\] Then $KR^{\ast}_{\Lambda_{E_7}(-\theta)} ((S^4)^{-\theta}) $ is isomorphic to \begin{align*}
        & KR^{\ast}_{\Lambda_{\mathbb{Z}/6}(4) }(S^0) \cong KR^{\ast}_{\Lambda_{\mathbb{Z}/6}(4) }(\pt) \oplus KR^{\ast}_{\Lambda_{\mathbb{Z}/6}(4) }(\pt) \\
      \cong   & KR^{\ast}(\pt)[x, q^{\pm }] /\langle x^6-q^4\rangle \oplus KR^{\ast}(\pt)[x, q^{\pm }] /\langle x^6-q^4\rangle.
    \end{align*}

\item For the conjugacy class represented by $r=\frac{1}{\sqrt{2}} (i+j)$, the Real centralizer \[ C^R_{\hat{E'_7}}(r)  \cong D_4.\] Thus, 
$KR^{\ast}_{\Lambda_{E_7}(r)} ((S^4)^{r}) $ is isomorphic to \begin{align*}
        & KR^{\ast}_{\Lambda_{\mathbb{Z}/4}(1) }(S^0) \cong KR^{\ast}_{\Lambda_{\mathbb{Z}/4}(1) }(\pt )  \oplus KR^{\ast}_{\Lambda_{\mathbb{Z}/4}(1) }(\pt)   \\
      \cong   & KR^{\ast}(\pt)[x, q^{\pm }] /\langle x^4-q\rangle \oplus KR^{\ast}(\pt)[x, q^{\pm }] /\langle x^4-q\rangle.
    \end{align*}

\item For the conjugacy class represented by $t= \frac{1}{\sqrt{2}}(1+i)$, its Real centralizer \[ C^R_{\hat{E'_7}}(t)  \cong D_8.\]
Thus, $KR^{\ast}_{\Lambda_{E_7}(t)} ((S^4)^{t}) $ is isomorphic to \begin{align*}
        & KR^{\ast}_{\Lambda_{\mathbb{Z}/8}(1) }(S^0) \cong KR^{\ast}_{\Lambda_{\mathbb{Z}/8}(1) }(\pt) \oplus KR^{\ast}_{\Lambda_{\mathbb{Z}/8}(1) }(\pt) \\
      \cong   & KR^{\ast}(\pt)[x, q^{\pm }] /\langle x^8-q\rangle \oplus KR^{\ast}(\pt)[x, q^{\pm }] /\langle x^8-q\rangle.
    \end{align*}

    \item For the conjugacy class represented by  $-t$,  its Real centralizer  \[ C^R_{\hat{E'_7}}(-t)  \cong D_8.\]
Thus, $KR^{\ast}_{\Lambda_{E_7}(-t)} ((S^4)^{-t} $ is isomorphic to \begin{align*}
        & KR^{\ast}_{\Lambda_{\mathbb{Z}/8}(1) }(S^0) \cong KR^{\ast}_{\Lambda_{\mathbb{Z}/8}(1) }(\pt) \oplus KR^{\ast}_{\Lambda_{\mathbb{Z}/8}(1) }(\pt) \\
      \cong   & KR^{\ast}(\pt)[x, q^{\pm }] /\langle x^8-q^5\rangle \oplus KR^{\ast}(\pt)[x, q^{\pm }] /\langle x^8-q^5\rangle.
    \end{align*}
    
\end{enumerate}

Thus, in conclusion, \begin{align*}
\QR^{\ast}(S^4\git E_7)  = &KR^{\ast}_{\Lambda_{E_7}(I)}((S^4)^I)  \times  KR^{\ast}_{\Lambda_{E_7}(-I)}(S^4)^{-I} \times KR^{\ast}_{\Lambda_{E_7}(i)} ((S^4)^i) \\
&\times 
 KR^{\ast}_{\Lambda_{E_7}(s)} ((S^4)^s)  \times  KR^{\ast}_{\Lambda_{E_7}(-s)} ((S^4)^{-s}) \times KR^{\ast}_{\Lambda_{E_7}(r)} ((S^4)^{r}) \\
 &\times    KR^{\ast}_{\Lambda_{E_7}(t)} ((S^4)^{t})   \times KR^{\ast}_{\Lambda_{E_7}(-t)} ((S^4)^{-t})\\
    \cong &KR^{\ast}_{E_7}(\pt)[q^{\pm}]
    \oplus  KR^{\ast}_{E_7}(\pt)[q^{\pm}] \\
    &\times    \prod_{1}^{2}KR^{\ast}_{T_{7}}(\pt) [q^{\pm}] \oplus KR^{\ast + \hat{\nu}_{\Lambda^R_{\hat{E_7'}(-I)}, sign}}_{T_{7}}(\pt) [q^{\pm}] \\
    &\times KR^{\ast}(\pt)[x, q^{\pm }] /\langle x^8-q^2\rangle \oplus KR^{\ast}(\pt)[x, q^{\pm }] /\langle x^8-q^2\rangle \\
    & \times KR^{\ast}(\pt)[x, q^{\pm }] /\langle x^6-q\rangle \oplus KR^{\ast}(\pt)[x, q^{\pm }] /\langle x^6-q\rangle \\
    & \times KR^{\ast}(\pt)[x, q^{\pm }] /\langle x^6-q^4\rangle \oplus KR^{\ast}(\pt)[x, q^{\pm }] /\langle x^6-q^4\rangle \\
    & \times  KR^{\ast}(\pt)[x, q^{\pm }] /\langle x^4-q\rangle \oplus KR^{\ast}(\pt)[x, q^{\pm }] /\langle x^4-q\rangle \\
    & \times KR^{\ast}(\pt)[x, q^{\pm }] /\langle x^8-q\rangle \oplus KR^{\ast}(\pt)[x, q^{\pm }] /\langle x^8-q\rangle \\
    & \times KR^{\ast}(\pt)[x, q^{\pm }] /\langle x^8-q^5\rangle \oplus KR^{\ast}(\pt)[x, q^{\pm }] /\langle x^8-q^5\rangle,
\end{align*}  where $sign$ is the sign representation of $\Z/2$.

\end{example}

\begin{example} 

In this example we compute $\QR^{\ast}(S^4\git E_8)$ where $E_8$ is the binary icosahedral group.  A presentation of this group is $$\langle r, s, t \mid (st)^2= s^3= t^5 = -1.\rangle. $$
The cardinality of $E_8$ is 120.
In this example, 
we use $\tau$ to denote $\frac{1+\sqrt{5}}{2}$ and   $\sigma$ to denote the number $\frac{1-\sqrt{5}}{2}$.
We take the Real structure $\hat{E_8}'$ on $E_8$, i.e. 
$E_8\langle s'\rangle$.  

By \cite[page 7635, Table 1]{KocaAKo} and direct computation, 
we obtain a list of the representatives of the conjugacy classes of $E_8$, the centralizers of each representative, whether it's fixed under the involution or not,  and the corresponding fixed point spaces in Figure \ref{E8:conj:c:fps}.

\begin{figure}
\begin{center}
\begin{tabular}{|c | c | c | c| } 
 \hline Representatives $\xi$  &Centralizers  & Fixed points under  & $(S^4)^{\xi}$\\
 of Conjugacy classes & $C_{E_8}(\xi)$ & the involution?  &
 \\ \hline
$1$ & $ E_8$ & Y& $S^4$\\
$-1$ & $E_8$ & Y &$ S^0$\\
$y_3:= \frac{1}{2}(\tau + i + \sigma k)$ & $\langle  y_3 \rangle \cong \mathbb{Z}/10$ &Y  & $S^0$\\
$y_4:= y_5^2 = \frac{1}{2}(-\tau + \sigma i-j)$ & $\langle  y_5 \rangle \cong \mathbb{Z}/10$ &Y & $S^0$\\
$y_5:= \frac{1}{2}(\sigma + i + \tau j)$ & $\langle  y_5 \rangle \cong \mathbb{Z}/10$ & Y & $S^0$ \\
$y_6:= y_3^2 = \frac{1}{2}(-\sigma + \tau i -k)$ & $\langle  y_3 \rangle \cong \mathbb{Z}/10 $  & Y & $S^0$ \\
$y_7:= \frac{1}{2}(1 + i +  j +k)$ & $\langle  y_7 \rangle \cong \mathbb{Z}/6$ & Y & $S^0$ \\
$y_8:= y_7^2=\frac{1}{2}(- 1  + i + j + k)$ & $\langle  y_7 \rangle \cong \mathbb{Z}/6$ & Y & $S^0$ \\
$y_9:= i $  &$ \langle  y_9 \rangle \cong \mathbb{Z}/4$ &Y  & $S^0$ \\
 \hline
\end{tabular} \caption{Conjugacy classes, centralizers and fixed point spaces} \label{E8:conj:c:fps}
\end{center}
\end{figure}


Next we compute each factor of $\QR^{\ast}(S^4\git E_8)$ corresponding to each conjugacy class of $ E_8$.
\begin{enumerate}
    \item For the conjugacy class $\{I\}$, the Real centralizer $C^R_{\hat{E_8}'}(I) = \hat{E_8}'$. Thus, by \cite[Theorem 5.1]{atiyah1968_Bott},
    \begin{align*}KR^{\ast}_{\Lambda_{E_8}(I)}((S^4)^I) &\cong KR^{\ast}_{E_8\times \mathbb{T}} (S^4) \cong KR^{\ast}_{E_8\times \mathbb{T}} (S^0)  \\
    &\cong KR^{\ast}_{E_8\times \mathbb{T}} (\pt) \oplus KR^{\ast}_{E_8\times \mathbb{T}} (\pt) \cong KR^{\ast}_{E_8} (\pt)[q^{\pm}] \oplus KR^{\ast}_{E_8} (\pt)[q^{\pm}].\end{align*}

\item For the conjugacy class $\{-I\}$, the Real centralizer $C^R_{\hat{E_8}'}(-I)= \hat{E_8}'$. Thus, by Corollary \ref{RM:Z2}, 
\begin{align*}
    KR^{\ast}_{\Lambda_{E_8}(-I)}(S^4)^{-I} &\cong KR^{\ast}_{\Lambda_{E_8}(-I)} (S^0) \\
    & \cong \prod_1^2 KR^{\ast}_{\Lambda_{E_8}(-I)} (\pt) \oplus KR^{\ast}_{\Lambda_{E_8}(-I)} (\pt)\\
    &\cong \prod_{1}^{2}KR^{\ast}_{T_{8}\times \mathbb{T}}(\pt) \oplus KR^{\ast + \hat{\nu}_{\Lambda^R_{\hat{E_8'}(-I)}, sign}}_{T_{8}\times \mathbb{T}}(\pt) \\ & \cong \prod_{1}^{2}KR^{\ast}_{T_{8}}(\pt) [q^{\pm}]\oplus KR^{\ast + \hat{\nu}_{\Lambda^R_{\hat{E_8'}(-I)}, sign}}_{T_{8}}(\pt)[q^{\pm}],
\end{align*} where $sign$ is the sign representation of $\Z/2$.

\item 
    For the conjugacy class represented by  $y_3$,  its  Real centralizer \[C^R_{\hat{E_8}}(y_3) \cong D_{10}.\]
    Thus, $KR^{\ast}_{\Lambda_{E_8}(y_3)} ((S^4)^{y_3}) $ is isomorphic to \begin{align*}
        & KR^{\ast}_{\Lambda_{\mathbb{Z}/10}(1) }(S^0) \cong KR^{\ast}_{\Lambda_{\mathbb{Z}/10}(1) }(\pt)  \oplus KR^{\ast}_{\Lambda_{\mathbb{Z}/10}(1) }(\pt) \\
      \cong   & KR^{\ast}(\pt)[x, q^{\pm }] /\langle x^{10}-q\rangle \oplus KR^{\ast}(\pt)[x, q^{\pm }] /\langle x^{10}-q\rangle.
    \end{align*}

\item 
    For the conjugacy class represented by $y_4$,   the Real centralizer \[C^R_{\hat{E_8}'}(y_4) \cong D_{10}.\] Thus,
 $KR^{\ast}_{\Lambda_{E_8}(y_4)} ((S^4)^{y_4}) $ is isomorphic to \begin{align*}
        & KR^{\ast}_{\Lambda_{\mathbb{Z}/10}(2) }(S^0) \cong KR^{\ast}_{\Lambda_{\mathbb{Z}/10}(2) }(\pt)   \oplus KR^{\ast}_{\Lambda_{\mathbb{Z}/10}(2) }(\pt)  \\
      \cong   & KR^{\ast}(\pt)[x, q^{\pm }] /\langle x^{10}-q^2\rangle \oplus KR^{\ast}(\pt)[x_2, q^{\pm }] /\langle x^{10}-q^2\rangle.
    \end{align*}

    \item     For the conjugacy class represented by $y_5$,  the Real centralizer $C^R_{\hat{E_8}'}(y_5) \cong D_{10}.$
    Thus,  the factor $KR^{\ast}_{\Lambda_{E_8}(y_5)} ((S^4)^{y_5}) $ is isomorphic to \begin{align*}
        & KR^{\ast}_{\Lambda_{\mathbb{Z}/10}(1) }(S^0) \cong KR^{\ast}_{\Lambda_{\mathbb{Z}/10}(1) }(\pt)  \oplus KR^{\ast}_{\Lambda_{\mathbb{Z}/10}(1) }(\pt)   \\
      \cong   & KR^{\ast}(\pt)[x, q^{\pm }] /\langle x^{10}-q \rangle \oplus KR^{\ast}(\pt)[x, q^{\pm }] /\langle x^{10}-q\rangle.
    \end{align*}

  \item      For the conjugacy class represented by $y_6$,  the Real centralizer  $C^R_{\hat{E_8}'}(y_6) \cong D_{10}.$ Thus, the factor 
 $KR^{\ast}_{\Lambda_{E_8}(y_6)} ((S^4)^{y_6}) $ is isomorphic to \begin{align*}
        & KR^{\ast}_{\Lambda_{\mathbb{Z}/10}(2) }(S^0) \cong  KR^{\ast}_{\Lambda_{\mathbb{Z}/10}(2) }(\pt)   \oplus  KR^{\ast}_{\Lambda_{\mathbb{Z}/10}(2) }(\pt)  \\
      \cong   & KR^{\ast}(\pt)[x, q^{\pm }] /\langle x^{10}-q^2\rangle \oplus KR^{\ast}(\pt)[x, q^{\pm }] /\langle x^{10}-q^2\rangle.
    \end{align*}

    \item     For the conjugacy class represented by $y_7$,  the Real centralizer  $C^R_{\hat{E_8}'}(y_7) \cong D_6.$ Thus, the factor 
  $KR^{\ast}_{\Lambda_{E_8}(y_7)} ((S^4)^{y_7}) $ is isomorphic to \begin{align*}
        & KR^{\ast}_{\Lambda_{\mathbb{Z}/6}(1) }(S^0) \cong KR^{\ast}_{\Lambda_{\mathbb{Z}/6}(1) }(\pt )  \oplus KR^{\ast}_{\Lambda_{\mathbb{Z}/6}(1) }(\pt) \\
      \cong   & KR^{\ast}(\pt)[x, q^{\pm }] /\langle x^{6}-q \rangle \oplus KR^{\ast}(\pt)[x, q^{\pm }] /\langle x^{6}-q\rangle.
    \end{align*}

  \item      For the conjugacy class represented by $y_8$, the Real centralizer  $C^R_{\hat{E_8}'}(y_8) \cong D_6.$ Thus, the factor 
  $KR^{\ast}_{\Lambda_{E_8}(y_8)} ((S^4)^{y_8}) $ is isomorphic to \begin{align*}
        & KR^{\ast}_{\Lambda_{\mathbb{Z}/6}(2) }(S^0) \cong KR^{\ast}_{\Lambda_{\mathbb{Z}/6}(2) }(\pt)   \oplus KR^{\ast}_{\Lambda_{\mathbb{Z}/6}(2) }(\pt)  \\
      \cong   & KR^{\ast}(\pt)[x, q^{\pm }] /\langle x^{6}-q^2\rangle \oplus KR^{\ast}(\pt)[x, q^{\pm }] /\langle x^{6}-q^2\rangle.
    \end{align*}

    \item     For the conjugacy class represented by $y_9$,   the Real centralizer \[C^R_{\hat{E_8}'}(y_9) \cong D_4.\]
    Thus, the corresponding factor $KR_{\Lambda_{E_8}(y_9)} ((S^4)^{y_9}) $ is isomorphic to \begin{align*}
        & KR^{\ast}_{\Lambda_{\mathbb{Z}/4}(1) }(S^0) \cong KR^{\ast}_{\Lambda_{\mathbb{Z}/4}(1) }(\pt ) \oplus KR^{\ast}_{\Lambda_{\mathbb{Z}/4}(1) }(\pt) \\
      \cong   & KR^{\ast}(\pt)[x, q^{\pm }] /\langle x^{4}-q \rangle \oplus KR^{\ast}(\pt)[x, q^{\pm }] /\langle x^{4}-q\rangle.
    \end{align*}

\end{enumerate}

In conclusion, \begin{align*}
    \QR(S^4 \git E_8) =& KR_{\Lambda_{E_8}(I)}((S^4)^I) \times KR_{\Lambda_{E_8}(-I)}((S^4)^{-I}) \times  KR_{\Lambda_{E_8}(y_3)} ((S^4)^{y_3}) \\ 
    &\times KR_{\Lambda_{E_8}(y_4)} ((S^4)^{y_4})  \times KR_{\Lambda_{E_8}(y_5)} ((S^4)^{y_5})  \times  KR_{\Lambda_{E_8}(y_6)} ((S^4)^{y_6}) \\
    &\times KR_{\Lambda_{E_8}(y_7)} ((S^4)^{y_7}) \times   KR_{\Lambda_{E_8}(y_8)} ((S^4)^{y_8}) \times  KR_{\Lambda_{E_8}(y_9)} ((S^4)^{y_9})\\
     \cong &  KR^{\ast}_{E_8} (\pt)[q^{\pm}] \oplus KR^{\ast}_{E_8} (\pt)[q^{\pm}]  \\
     &\times  \prod_{1}^{2}KR^{\ast}_{T_{8}}(\pt) [q^{\pm}]\oplus KR^{\ast + \hat{\nu}_{\Lambda^R_{\hat{E_8'}(-I)}, sign}}_{T_{8}}(\pt)[q^{\pm}] \\ &\times 
      KR^{\ast}(\pt)[x, q^{\pm }] /\langle x^{10}-q\rangle \oplus KR^{\ast}(\pt)[x, q^{\pm }] /\langle x^{10}-q\rangle \\
      &\times  KR^{\ast}(\pt)[x, q^{\pm }] /\langle x^{10}-q^2\rangle \oplus KR^{\ast}(\pt)[x, q^{\pm }] /\langle x^{10}-q^2\rangle\\
      &\times  KR^{\ast}(\pt)[x, q^{\pm }] /\langle x^{10}-q \rangle \oplus KR^{\ast}(\pt)[x, q^{\pm }] /\langle x^{10}-q\rangle\\
      &\times  KR^{\ast}(\pt)[x, q^{\pm }] /\langle x^{10}-q^2\rangle \oplus KR^{\ast}(\pt)[x, q^{\pm }] /\langle x^{10}-q^2\rangle\\
      &\times KR^{\ast}(\pt)[x, q^{\pm }] /\langle x^{6}-q \rangle \oplus KR^{\ast}(\pt)[x, q^{\pm }] /\langle x^{6}-q\rangle\\
      &\times  KR^{\ast}(\pt)[x, q^{\pm }] /\langle x^{6}-q^2\rangle \oplus KR^{\ast}(\pt)[x, q^{\pm }] /\langle x^{6}-q^2\rangle\\
      &\times  KR^{\ast}(\pt)[x, q^{\pm }] /\langle x^{4}-q \rangle \oplus KR^{\ast}(\pt)[x, q^{\pm }] /\langle x^{4}-q\rangle,
\end{align*}  where $sign$ is the sign representation of $\Z/2$.

\end{example}

\begin{remark}

    As we can see in the examples of this section, most computation  lead to the equivariant KR-theory of a single point. 
    The whole data of the equivariant KR-theory, by the computation in \cite[Section 8]{atiyah1969} and \cite[Proposition 3.1]{Fok2013THERK}, is given as
    \begin{align*} & KR^{\ast}_G(\pt) := \sum_{n=0}^{7} KR^{-n}_G(\pt) \\
    = &RR(G) \oplus RR(G)/ \rho(R(G)) \oplus R(G)/ j(RH(G)) \oplus 0 \\
    &\oplus RH(G)\oplus RH(G)/\eta(R(G))
    \oplus R(G)/i(RR(G)) \oplus 0\end{align*} where $i: RR(G)\rightarrow R(G)$ and $j: RH(G)\rightarrow R(G)$ are the forgetful functors, the map $\rho$ is given explicitly in \cite[Proposition 2.17]{Fok2013THERK} and the map $\eta$ is given explicitly in \cite[Proposition 2.24]{Fok2013THERK}. 

    In addition,   there is a graded ring isomorphism (see \cite[Section 8]{atiyah1969})
\[
KR^{\ast}(\pt)
\cong
\Z[\eta, \mu] \slash \langle 2 \eta, \eta^3, \eta \mu, \mu^2 -4 \rangle,
\qquad
\deg \eta =-1, \;\; \deg \mu =-4.
\]
\end{remark}

\section{Quasi-elliptic cohomology of $S^4$ acted by a finite subgroup of $\mathrm{Spin}(4)$} \label{prod_rc_qell}

In this section we compute $\QR^{\ast}_{G}(S^4)$ with $G$ a finite subgroup of $\mathrm{Spin}(4)$. The $\mathrm{Spin}(4)$-action on $S^4$ that we are interested in is that given by the formulas \eqref{spin(4)action} and \eqref{spins5}. 

Denote by $\mathbb{H} \simeq_{\mathbb{R}} \mathbb{R}^4$ the space of quaternions, to be regarded mainly as a real module under quaternion multiplication from the left and right, in particular by unit quaternions
$$
  q \in \mathbb{H}
  \;\;\;\;
     \vdash
  \;\;\;\;
  q \cdot q^\ast = 1
  \;\;
    \Leftrightarrow
  \;\;
  q \,\in\, S(\mathbb{H})
  \,.
$$
 
  We have group isomorphism
  $$
    \mathrm{Spin}(3)
    \;\simeq\;
    S(\mathbb{H})
  $$
  and 
  $$
    \mathrm{Spin}(4) 
    \;\simeq\;
    \mathrm{Spin}(3)
    \times
    \mathrm{Spin}(3)
  $$
  under which the spin double cover of $\mathrm{SO}(4)$ is given by
  \begin{equation}\label{spin(4)_Act}
    \begin{tikzcd}[
      row sep=0pt,
      column sep=0pt
    ]
      \mathrm{Spin}(4)
      \ar[
        r,
        phantom,
        "{\simeq}"
      ]
      &[8pt]
      \mathrm{Spin}(3)
      \times
      \mathrm{Spin}(3)
      \ar[rr, ->>]
      &&
      \mathrm{SO}(\mathbb{H})
      \ar[
        r,
        phantom,
        "{ \simeq }"
      ]
      &[-12pt]
      \mathrm{SO}(4)
      \\
      &
      (e_1, e_2)
      &\mapsto&
      \big(
        q
        \,\mapsto\,
        e_1 \cdot q \cdot e_2^\ast
      \big)
    \end{tikzcd}
  \end{equation}

\subsection{Warm-up Examples}

We start with a simple example. 

\begin{example} \label{factor_dual_spin5}
In \cite[Section 6]{huan2020_v2} and Section \ref{Real_spin3_s4} we compute complex and Real quasi-elliptic cohomology of $S^4$ under the action of  the finite subgroups of $\mathrm{Spin}(3)\times 1\subset \mathrm{Spin}(4) \subset \mathrm{Spin}(5)$. In this example We consider the "dual" of them, i.e. the finite subgroup of $1\times \mathrm{Spin}(3) \subset \mathrm{Spin}(4)\subset \mathrm{Spin}(5)$, which are the groups \[1\times G_n, 1\times 2D_{2n}, 1\times E_6, 1\times E_7, 1\times E_8.\]

For a point $(1, r)\in 1\times \mathrm{Spin}(3)$, it acts on a point $y\in \mathbb{H}$  by \[ (1, r) \cdot y = y \overline{r} = \overline{r\overline{y}}.\] 
For any finite subgroup $G$ of $1\times \mathrm{Spin}(3)$, for any torsion point $(1, r)\in G$, $(S^4)^{(1, r)} = \overline{(S^4)^{(r, 1)}}$; and the centralizer $C_{1\times G}(1,r) = 1\times C_G(r) \cong C_G(r) \times 1= C_{G\times 1}(r, 1)$. Thus, $\Lambda_{1\times G}(1,r) = \Lambda_{G\times 1}(r, 1)$. 
For the Real case, the Real centralizer $C^R_{1\times \hat{G}}(1,r) = 1\times C^R_{\hat{G}}(r) \cong C^R_{\hat{G}}(r) \times 1= C^R_{\hat{G}\times 1}(r, 1)$
It's straightforward to check case by case that \[QEll^{\ast}_{1\times G}(S^4) \cong QEll^{\ast}_{G\times 1}(S^4)\] and the Real quasi-elliptic cohomology \[\QR^{\ast}_{1\times G}(S^4) \cong \QR^{\ast}_{G\times 1}(S^4). \]

\end{example}

\begin{example} \label{refl_S^4}

In this example we study the $\mathbb{Z}/2$-action on $S^4$ induced by the involution $x$ on $\mathbb{H}$
\[x: a+bi+cj+dk\mapsto (-a)+bi+cj+dk.\] The north pole and south pole are both fixed points under the involution.

There are two conjugacy classes in $\Z /2=\{1, \tau\}$ corresponding to its two elements.

Below we compute the factors of $QEll_{\Z/2}(S^4)$.
\begin{itemize}
    \item For the conjugacy class $1$, $(S^4)^1$ is $S^4$ itself. $\Lambda_{\Z/2}( 1) \cong \Z/2 \times \mathbb{T}$. \begin{align*} K_{\Lambda_{\Z/2}(1)}(S^4)^1 &\cong K_{\Z/2\times \mathbb{T}}(S^4) \cong K^{\ast}_{\Z/2\times \mathbb{T}}(S^0) \\
    &\cong \Z[x, q^{\pm }]/\langle x^2-1\rangle \oplus \Z[y, q^{\pm }]/\langle y^2-1\rangle. 
    \end{align*}

    \item For the conjugacy class $\tau$, $(S^4)^{\tau} = \{bi+cj+dk \in \mathbb{H} \mid b, c, d\in \mathbb{R}\}\cup \{\infty\} \cong S^3.$
    \begin{align*}
        K_{\Lambda_{\Z/2}(\tau)}(S^4)^{\tau} &\cong K_{\Lambda_{\Z/2}(\tau)}(S^3) \cong K_{\Lambda_{\Z/2}(\tau)}(S^0) \\
        &\cong \Z[x, q^{\pm}]/\langle x^2-q\rangle \oplus \Z[y, q^{\pm}]/\langle y^2-q\rangle.
    \end{align*}
    
\end{itemize}

Next we compute $\QR^{\ast}_{\Z/2}(S^4)$. If we take the Real structure on $\Z/2$ to be the Dihedral Real structure.
We can take the reflection to be \[y: \mathbb{H}\longrightarrow \mathbb{H}, \quad (a+bi+cj+dk) \mapsto (-a-bi-cj-dk).\] The composition $x\circ y$ sends a point $a+bi+cj+dk$ to $a-bi-cj-dk$, i.e. the quaternion conjugation. The group generated by $x$ and $y$ is the dihedral group $D_4$.

And the Real centralizers $C^R_{D_4}(\alpha) = D_4$ for  $\alpha= 1$, $\tau$ in $\Z/2$.
The factors of $\QR^{\ast}_{\Z/2}(S^4)$ is computed below.
\begin{itemize}
    \item For the conjugacy class $1$, $\Lambda_{\Z/2}(1)\cong \Z/2 \times \mathbb{T}$.  
    \begin{align*}
        KR^{\ast}_{\Lambda_{\Z/2}(1)}(S^4)^1 &\cong KR^{\ast}_{\Lambda_{\Z/2}(1)}(S^4) \cong KR^{\ast}_{\Lambda_{\Z/2}(1)}(S^0) \\
        &\cong  KR^{\ast}_{\Lambda_{\Z/2}(1)}(\pt ) \oplus KR^{\ast}_{\Lambda_{\Z/2}(1)}(\pt ) \\
        &\cong KR^{\ast}(\pt)[x, q^{\pm}]/\langle x^2-1 \rangle \oplus  KR^{\ast}(\pt)[y, q^{\pm}]/\langle y^2-1 \rangle. 
    \end{align*}

    \item 
    For the conjugacy class $\tau$,   
    \begin{align*}
 KR^{\ast}_{\Lambda_{\Z/2}(\tau)}(S^4)^\tau &\cong KR^{\ast}_{\Lambda_{\Z/2}(\tau)}(S^3) \cong  KR^{\ast}_{\Lambda_{\Z/2}(\tau)}(S^0) \\
 &\cong KR^{\ast}_{\Lambda_{\Z/2}(\tau)}(\pt ) \oplus KR^{\ast}_{\Lambda_{\Z/2}(\tau)}(\pt) \\
 &\cong KR^{\ast}(\pt)[x, q^{\pm}]/\langle x^2-q \rangle \oplus KR^{\ast}(\pt)[y, q^{\pm}]/\langle y^2-q \rangle.
    \end{align*}
\end{itemize}

\end{example}

Next we compute $QEll_G(S^4)$ and $\QR^{\ast}_{G}(S^4)$ with $G$ a cyclic subgroup of $\mathrm{Spin}(4)$.

\begin{example}\label{quaternion_cyclic}

Let \[ G=\langle \left[ \begin{array}{cc} e^{\frac{2\pi i m_1}{n_1}} & 0 \\  0 & e^{\frac{2\pi i m_2}{n_2}}
\end{array} \right] \in U(2, \mathbb{H}) \mid  m_1, m_2\in \Z \rangle. \]  Let  \[\alpha:=\left[ \begin{array}{cc} e^{\frac{2\pi i p_1}{n_1}} & 0 \\  0 & e^{\frac{2\pi i p_2}{n_2}}
\end{array} \right]\] denote a generator of the cyclic group. We can assume that $p_1$ and  $n_1$ are coprime, and $p_2$ and $n_2$ are coprime. 
The order of $G$ is the least common multiple $N$ of $n_1$ and $n_2$. 

Then for any $\alpha^m \in G$,  the centralizer \[C_G(\alpha^m)= G. \]
And \[(S^4)^{\alpha^m} = \begin{cases}
    S^4, &\text{  if  } \alpha^m =  I; \\
    S^0, &\text{  otherwise }.
\end{cases}\]

The group $G =  \langle \alpha \rangle$ is isomorphic to $\Z/N$. Then we can apply the results in \cite{huan2020_v2} and Example \ref{G_nDihedralS^4} directly. 

The complex quasi-elliptic cohomology is 
\begin{align*}
    QEll_{G}(S^4) &= \prod^N_{m=0}  K_{\Lambda_{G}(\alpha^m)}((S^4)^{\alpha^m} ) \\
    &\cong \prod^N_{m=0} \mathbb{Z}[q^{\pm}, x]/\langle x^N-q^m \rangle \oplus \mathbb{Z}[q^{\pm}, x]/\langle x^N-q^m \rangle.
\end{align*}

The Real quasi-elliptic cohomology is 
\[\QR^{\ast}(S^4\git G) \cong   \prod_{m=0}^{N-1} KR^{\ast}(\pt)[x, q^{\pm}] \slash \langle x^N-q^m \rangle \oplus  KR^{\ast}(\pt)[x, q^{\pm}] \slash \langle x^N-q^m \rangle.\]

\end{example}

\subsection{Product of finite subgroups}
 I didn't find all the finite subgroups of $\mathrm{Spin}(5)$ that have a well-defined action on $\mathbb{H}$. I will discuss some finite subgroups of the form $H\times K $ where both $H$ and $K$ are finite subgroups of $\mathrm{Spin}(3)$. 
 \begin{example}
 For any $(h, k)\in H\times K$, and $y\in \mathbb{H}$, as given in \eqref{spin(4)_Act}, \[ (h, k)\cdot y :=  h y \overline{k}.\]
The set of conjugacy classes $\pi_0(H\times K \git H\times K)$
 is one-to-one correspondent to $\pi_0(H\git H) \times \pi_0(K\git K)$.
In addition, \begin{equation}\Lambda_{H\times K}( h, k) \cong \Lambda_H(h) \times_{\mathbb{T}} \Lambda_K(k). \label{Lambd_prod}\end{equation}

If $(\widehat{H}, \pi_H)$ is a Real structure on $H$ and 
$(\widehat{K}, \pi_K)$ is a Real structure on $K$, then we have the product Real structure $$(\widehat{H}\times_{\Z/2} \widehat{K}, \pi )$$ where the projection \[\pi= \pi_H\times_{\Z/2} \pi_K: \widehat{H}\times_{\Z/2} \widehat{K}\longrightarrow \Z/2\] sends $(h, k)$ to $\pi_H(h) = \pi_K(k)$. For the Real centralizers, \[C^R_{\widehat{H}\times_{\Z/2} \widehat{K}}( h, k) \cong C^R_{\widehat{H}}(h) \times_{\Z/2} C^R_{\widehat{K}}(k). \] Thus,
\begin{equation}\Lambda^R_{\widehat{H}\times_{\Z/2} \widehat{K}}( h, k) \cong \Lambda^R_{\widehat{H}}(h) \times_{O(2)} \Lambda^R_{\widehat{K}}(k), 
\label{lambdR_prod}\end{equation} 
where $O(2)$ is the 2-dimensional orthogonal group.

\bigskip

In addition, if the  reflection in $\widehat{H}$ and $\widehat{K}$ on $\R^4$ are represented by the same matrix $\alpha \in U(2)$ with $\alpha^2 = I$, it defines a $\C$-linear map 
\begin{align*}
SU(2) &\longrightarrow SU(2) \\
A &\mapsto \alpha A \overline{\alpha}.
\end{align*}
Then, by direct computation, if we take $\alpha$ to be the reflection $s$ defined in Example \ref{Real_str_SU(2)}, the resulting reflection on $\mathbb{H}\cong \R^4$ is defined by  \[(a+bi+cj+dk) \mapsto (a-bi-cj+dk). \]
And if we take $\alpha$ to be the reflection $s'$ defined in Example \ref{Real_str_SU(2)}, the resulting reflection is \[ (a+bi+cj+dk)
\mapsto (a+bi-cj-dk).\]
In addition, if we take the reflection in $\widehat{H}$ to be $s$ and that on $\widehat{K}$ to be $s'$, the resulting reflection on $\mathbb{H}$ is 
\[ (a+bi+cj+dk) \mapsto (-c+di-aj+bk).\] And if we take the reflection in $\widehat{H}$ to be $s'$ and that on $\widehat{K}$ to be $s$, the resulting reflection on $\mathbb{H}$ is \[ (a+bi+cj+dk) \mapsto (c+di+aj+bk).\]

\end{example}

In fact, we have a conclusion generalizing Example \ref{factor_dual_spin5}.

\begin{proposition}
    Let $H$ and $K$ denote two finite subgroups of $\mathrm{Spin}(3)$. The product $H\times K$ acts on $S^4$ in the way as in \eqref{spin(4)_Act}. Then \[QEll^{\ast}_{H\times K}(S^4) \cong QEll^{\ast}_{K\times H}(S^4).\]
    Moreover, if $(\widehat{H}, \pi_H)$ is Real structure on $H$ and 
$(\widehat{K}, \pi_K)$ is Real structure on $K$, then, \[\QR^{\ast}_{H\times K}(S^4) \cong \QR^{\ast}_{K\times H}(S^4).\]
\end{proposition}

\begin{proof}
The factors of both $QEll^{\ast}_{H\times K}(S^4)$ and $QEll^{\ast}_{K\times H}(S^4)$ go through the set $\pi_0(H\git H) \times \pi_0(K\git K)$.

By \eqref{Lambd_prod}, for any $  \sigma\in H$, and $ \tau\in K$, \[\Lambda_{H\times K}( \sigma, \tau) \cong \Lambda_H(\sigma) \times_{\mathbb{T}} \Lambda_K(\tau) \cong \Lambda_K(\tau) \times_{\mathbb{T}}
\Lambda_H(\sigma) \cong \Lambda_{K\times H}(\tau,\sigma). \]

For any fixed point $a+bi+cj+dk \in \mathbb{H}$ of $(\sigma, \tau)$, we have the equality \[
\sigma(a+bi+cj+dk)\overline{\tau} = a+bi+cj+dk.
\] Taking the complex conjugate of both sides, we get \[ \tau (a-bi-cj-dk) \overline{\sigma}  = a-bi-cj-dk.\] Thus, the complex conjugate of the quaternion induces a one-to-one correspondence \[ (S^4)^{(\sigma, \tau)} \buildrel{\overline{(-)}}\over\longrightarrow (S^4)^{(\tau, \sigma)}.\]
Moreover, it is direct to show that for any element $(u, v)\in C_{H\times K}( \sigma, \tau) $, any $x= a+bi+cj+dk \in (S^4)^{(\sigma, \tau)}$, we have the equality
\[\overline{(u, v)\cdot x} = (v, u)\cdot \overline{x}.\] Note that $(v, u)\in C_{K\times H}(\tau, \sigma)$.
This leads to the isomorphism \[K^{\ast}_{\Lambda_{H\times K}( \sigma, \tau) }(S^4)^{(\sigma, \tau)}  \cong K^{\ast}_{\Lambda_{K\times H}(\tau,\sigma)} (S^4)^{(\tau, \sigma)}. \] 

Thus, \[QEll^{\ast}_{H\times K}(S^4) \cong QEll^{\ast}_{K\times H}(S^4).\]

\bigskip

For the Real case, the factors of both $\QR^{\ast}_{H\times K}(S^4)$ and $\QR^{\ast}_{K\times H}(S^4)$  go through the same set
$\pi_0((H\times K) \git_R (\widehat{H}\times_{\Z/2} \widehat{K}))$.
 In addition, by \eqref{lambdR_prod}, \[ \Lambda^R_{\widehat{H}\times_{\Z/2} \widehat{K}}( \sigma, \tau)  \cong \Lambda^R_{\widehat{K}\times_{\Z/2} \widehat{H}}( \tau, \sigma). \] And the complex conjugate \[ (S^4)^{(\sigma, \tau)} \longrightarrow (S^4)^{(\tau, \sigma)}\] commutes with the reflections, as shown below.
 \[ \xymatrix{ (S^4)^{(\sigma, \tau)} \ar[r]^{\overline{ (-)}} \ar[d]_{(s_H, s_K)}
 &(S^4)^{(\tau, \sigma)} \ar[d]^{(s_K, s_H)} \\ (S^4)^{(s_H\sigma s_H^{-1}, s_K\tau s_K^{-1})} \ar[r]^{\overline{(-)}} & (S^4)^{( s_K\tau s_K^{-1}, s_H\sigma s_H^{-1})}}\]
where $s_H$ is the reflection in $\widehat{H}$ and $s_K$ is the reflection in $\widehat{K}$.

Thus, we get  \[\QR^{\ast}_{H\times K}(S^4) \cong \QR^{\ast}_{K\times H}(S^4).\]
    
\end{proof}

\begin{proposition} \label{prod_QEll}

    Let $H$ and $K$ denote two finite subgroups of $\mathrm{Spin}(3)$. The product $H\times K$ acts on $S^4$ by the action given in \eqref{spin(4)action}. Let $(\widehat{H}, \pi_H)$ denote a Real structure on $H$ and 
$(\widehat{K}, \pi_K)$ a Real structure on $K$. Then we have the conclusions below.
\begin{enumerate}
 \item     
    The factor in $QEll_{H\times K}(S^4)$ corresponding to the conjugacy class $(h, k)$, i.e. $K_{\Lambda_{H\times K}(h,k)} (S^4)^{(h, k)}$, is isomorphic to 
\[\prod\limits_1^2 R(\Lambda_H(h))\otimes_{\Z[q^{\pm}]} R(\Lambda_K(k)).\]
     Then we have the isomorphism \begin{align*} QEll_{H\times K}(S^4) &= 
     \prod_{(h, k)\in \pi_0((H\times K)\git (H\times K))} K_{\Lambda_{H\times K}(h,k)} (S^4)^{(h, k)}  \\
     & \cong  \prod_{h\in \pi_0(H\git H), k\in \pi_0(K\git K)}  \prod\limits_1^2 R(\Lambda_H(h))\otimes_{\Z[q^{\pm}]} R(\Lambda_K(k)).\end{align*}
     
\item  The  factor in $\QR^{\ast}_{H\times K}(S^4)$ corresponding to the Real conjugacy class $(h, k)$, i.e. ${^{\pi}}K^{\ast}_{\Lambda_{H\times K}(h,k)} (S^4)^{(h, k)}$, is isomorphic to: 
\begin{itemize}
    \item \[\prod\limits_1^2 KR^{\ast}_{\Lambda_{H}(h)} (\pt)\otimes_{KR^\ast_{\mathbb{T}}(\pt)} KR^{\ast}_{\Lambda_{  K}(k)} (\pt) ,\]  if $(h, k) $ is a fixed point under the involution;
    \item \[\prod\limits_1^2 K^{\ast}_{\Lambda_H(h)}(\pt)\otimes_{\Z[q^{\pm}]} K^{\ast}_{\Lambda_K(k)}(\pt)  ,\] if $(h, k)$ is a free point under the involution.
\end{itemize}

\end{enumerate}
     
\end{proposition}

\begin{proof}
We prove the conclusion one by one.
\begin{enumerate}
    \item Note that \eqref{spin(4)action} defines a $4$-dimensional representation  of $H\times K$.
     Thus, $(S^4)^{(h, k)}$ is a representation sphere of $\Lambda_{H\times K}(h, k)$ and contains $S^0$ as a subspace. Whatever $(S^4)^{(h, k)}$ is, by \cite[Theorem 4.3]{atiyah1968_Bott}, we have \[K_{\Lambda_{H\times K}(h,k)} (S^4)^{(h, k)} \cong K_{\Lambda_{H\times K}(h, k)} (S^0).\] And the right hand side is isomorphic to \begin{align*} & K_{\Lambda_{H\times K}(h, k)} (\pt )\oplus K_{\Lambda_{H\times K}(h, k)} (\pt)\cong R(\Lambda_{H\times K}(h, k)) \oplus R(\Lambda_{H\times K}(h, k)) \\
    \cong & R(\Lambda_H(h))\otimes_{\Z[q^{\pm}]} R(\Lambda_K(k))\oplus R(\Lambda_H(h))\otimes_{\Z[q^{\pm}]} R(\Lambda_K(k)). \end{align*}

    \item The proof is similar to the complex case.
    Since $(S^4)^{(h, k)}$ is a Real representation sphere of $\Lambda^R_{\widehat{H}\times_{\Z/2} \widehat{K}} (h,k)$, as well as a complex representation sphere of $\Lambda_{H\times K}(h, k)$, thus, by  \cite[Theorem 4.3, Theorem 5.1]{atiyah1968_Bott}, the Freed-Moore K-theory ${^{\pi}}K^{\ast}_{\Lambda_{H\times K}(h,k)} (S^4)^{(h, k)}$ is isomorphic to \[{^{\pi}}K^{\ast}_{\Lambda_{H\times K}(h,k)} (S^0) \cong {^{\pi}}K^{\ast}_{\Lambda_{H\times K}(h,k)} (\pt) \oplus {^{\pi}}K^{\ast}_{\Lambda_{H\times K}(h,k)} (\pt) .\]
    In addition, 
    \[{^{\pi}}K^{\ast}_{\Lambda_{H\times K}(h,k)} (\pt)  = \begin{cases}
        KR^{\ast}_{\Lambda_{H\times K}(h,k)} (\pt), &\text{  if  } (h, k) \text{  is a fixed point under the involution;}\\
        K^{\ast}_{\Lambda_{H\times K}(h,k)} (\pt), &\text{   if  }(h, k) \text{    is   a free point under the involution.}
    \end{cases}\] 
And $KR^{\ast}_{\Lambda_{H\times K}(h,k)} (\pt) \cong KR^{\ast}_{\Lambda_{H}(h)} (\pt)\otimes_{KR^\ast_{\mathbb{T}}(\pt)} KR^{\ast}_{\Lambda_{  K}(k)} (\pt)$.
    
\end{enumerate} 
Then we get the conclusion immediately.

\end{proof}

\begin{remark}
    One probably subtle point is that, as indicated in \cite{huanyoung2022}, the Real structure we takes in the $\R$ in the general definition of the enhanced Real stabilizer
\[
\Lambda_{\hat{G}}^R(g) \simeq (\R \rtimes_{\pi} \hat{G}) \slash \langle (-1,g) \rangle.
\] is the reflection of $r\mapsto -r$. This coincides with the dihedral Real structure on $\mathbb{T}$. More explicitly, 
the involution defined from the dihedral Real structure on $\mathbb{T}$ is given by $t\mapsto -t$, which is the quotient of the reflection on $\R$.  

The Real representation ring $RR(\mathbb{T})$ for $\mathbb{T}$ with the dihedral Real structure, i.e. $O(2)$, is exactly $RR(\mathbb{T}; \R)$, which is isomorphic to $\Z[q^{\pm}]$. Thus, the isomorphism \[\Lambda_{\hat{G}\times_{\Z/2} \hat{H}}^R(g, h) \cong  \Lambda_{\hat{G}}^R(g) \times_{O(2)} \Lambda_{\hat{H}}^R(h) \] gives us the isomorphism of Real representation rings, i.e. \[RR(\Lambda_{G\times H}(g, h)) \cong  RR(\Lambda_{G}(g)) \otimes_{RR\mathbb{T}} RR(\Lambda_{H}(h)). \]
\end{remark}

\begin{example} \label{quaternion_cyclic_2}

In this example we compute $QEll_{G_n\times G_m}(S^4)$  
 with   \[ G_n= \{ e^{\frac{2\pi i k}{n}} \in \mathbb{H} \mid  k \in \mathbb{Z}\}; \qquad G_m = \{ e^{\frac{2\pi i j}{m}} \in \mathbb{H} \mid  j \in \mathbb{Z}\}. \] 

By \cite[Example 6.3]{huan2020_v2} and Proposition \ref{prod_QEll}(i), \begin{align*}
&QEll_{G_n\times G_m}(S^4) \cong \prod_{k=0}^{n-1}  \prod_{l = 0}^{m-1} K_{\Lambda_{G_n\times G_m} (k, l)}(S^4)^{(k,l)} \\
\cong & \prod_{k=0}^{n-1}  \prod_{l = 0}^{m-1} \Z[x_1, x_2, q^{\pm}]/\langle x_1^n-q^k, x_2^m-q^l \rangle \oplus \Z[y_1, y_2, q^{\pm}]/\langle y_1^n-q^k, y_2^m-q^l \rangle .
    \end{align*}

  We take the Real structure $\hat{G}'_n$ as defined in Example \ref{Real:G:SU(2)}, i.e. the group below together with the determinant map $\det$
\[ \langle G_n, \left[ {\begin{array}{cc} 1 & 0 \\  0 & -1 
\end{array}} \right] \rangle. \] It is isomorphic to the dihedral group $D_{2n}$. 
As discussed in Example \ref{G_nDihedralS^4}, all the elements in $\pi_0(G_n \git_R \hat{G}'_n)$ and $\pi_0(G_m \git_R \hat{G}'_m)$ are fixed points under the involution, thus, so are those in $\pi_0(G_n\times G_m \git_R \widehat{G_n\times_{\Z/2} G_m})$.

By Example \ref{G_nDihedralS^4}  and  Proposition \ref{prod_QEll}(ii), \begin{align*}
   & \QR^{\ast}_{G_n\times G_m}(S^4) \cong     \prod_{k=0}^{n-1}  \prod_{l = 0}^{m-1} KR^{\ast}_{\Lambda_{G_n\times G_m} (k, l)}(S^4)^{(k,l)}  \\
  \cong &    \prod_{k=0}^{n-1}  \prod_{l = 0}^{m-1} KR^{\ast}(\pt)[x_1, x_2, q^{\pm}]\slash \langle x_1^n - q^k, x_2^m -q^l \rangle \oplus  KR^{\ast}(\pt)[y_1, y_2, q^{\pm}]\slash \langle y_1^n - q^k, y_2^m -q^l \rangle .
    \end{align*}
    
\end{example}

\begin{example}
Let $n$ and $m$ be positive integers.  Let $G_n < \mathrm{Spin}(3)$ denote the cyclic group \[ \{ e^{\frac{2\pi i k}{n}} \in \mathbb{H} \mid  k \in \mathbb{Z}\} \]
 and  $2D_{2m}$ denote the binary Dihedral group \[  \langle G_{2m}, \left[ {\begin{array}{cc} 0 & -1 \\  1 &  0 
\end{array}} \right]\rangle < \mathrm{Spin}(3). \] 
In this example we compute $QEll^{\ast}_{G_n\times 2D_{2m}}(S^4)$ and $\QR^{\ast}_{G_n\times 2D_{2m}}(S^4)$.

Let $\tau$ denote $\left[ {\begin{array}{cc} 0 & -1 \\  1 &  0 
\end{array}} \right]$ in $2D_{2m}$, which is $-j$ in term of quaternions. 

The factors of $QEll_{G_n\times 2D_{2m}} (S^4)$ corresponding to each conjugacy class is computed one by one below.  We first compute the factors  corresponding to the conjugacy classes represented by \begin{equation} \label{alpha_matrix}\alpha:= \left[ {\begin{array}{cc} e^{\frac{2\pi i k}{n}} & 0 \\  0 &  e^{\frac{2\pi i p}{2m}} 
\end{array}} \right] \in U(2, \mathbb{H}). \end{equation} 

\begin{enumerate}
    
\item 
If $\alpha = I$, \begin{align*}K_{\Lambda_{G_n\times 2D_{2m}}(I)} (S^4)^{I} &= K_{\Lambda_{G_n\times 2D_{2m}}(I )} (S^4)  \buildrel{(\ast)}\over\cong K_{\Lambda_{G_n\times 2D_{2m}}(I)} (S^0) \cong K_{G_n\times 2D_{2m}}(S^0)\otimes \Z[q^{\pm}] \\
&\cong (R(G_n\times 2D_{2m})\oplus R(G_n\times 2D_{2m}))\otimes \Z[q^{\pm}] \\ &\cong R(2D_{2m})[x_1, x_2, q^{\pm}]/\langle x_1^n-1, x_2^n-1 \rangle 
\end{align*} where  the isomorphism $(\ast)$ is by \cite[Theorem 4.3]{atiyah1968_Bott}.

\item 

If the $e^{\frac{2\pi i p}{2m}}$ in \eqref{alpha_matrix} is not $\pm I$, the centralizer $C_{G_n\times 2D_{2m}}(\alpha) = G_n\times G_{2m}$. \begin{align*}
K_{\Lambda_{G_n\times 2D_{2m}}(\alpha)} (S^4 )^{\alpha} &\cong K_{\Lambda_{G_n\times 2D_{2m}}(\alpha)} (S^0 ) \cong R(\Lambda_{G_n\times 2D_{2m}}(\alpha))\oplus R(\Lambda_{G_n\times 2D_{2m}}(\alpha))  \\ &\cong \Z[x_1, x_2, q^{\pm}]/\langle x_1^n- q^{k}, x_2^{2m}-q^{p} \rangle \oplus \Z[x_1, x_2, q^{\pm}]/\langle x_1^n- q^{k}, x_2^{2m}-q^{p} \rangle.\end{align*}

\item 
If  the $e^{\frac{2\pi i p}{2m}}$ in \eqref{alpha_matrix} is  $\pm I$, the centralizer $C_{G_n\times 2D_{2m}}(\alpha) = G_n\times 2D_{2m}$. 

\begin{itemize}
\item If $e^{\frac{2\pi i p}{2m}}= I$,
\begin{align*}
K_{\Lambda_{G_n\times 2D_{2m}}(\alpha)} (S^4 )^{\alpha} &\cong K_{\Lambda_{G_n\times 2D_{2m}}(\alpha)} (S^0 ) \cong R(\Lambda_{G_n\times 2D_{2m}}(\alpha))\oplus R(\Lambda_{G_n\times 2D_{2m}}(\alpha))  \\ &\cong R(2D_{2m})[x, q^{\pm}]/\langle x^n- q^{k}  \rangle \oplus  R(2D_{2m})[x', q^{\pm}]/\langle x'^n- q^{k}  \rangle .\end{align*}

\item If $e^{\frac{2\pi i p}{2m}}= -I$,  Applying Lemma \ref{dcld2}, we get \begin{align*}
K_{\Lambda_{G_n\times 2D_{2m}}(\alpha)} (S^4 )^{\alpha} &\cong K_{\Lambda_{G_n\times 2D_{2m}}(\alpha)} (S^0 ) \cong R(\Lambda_{G_n\times 2D_{2m}}(\alpha))\oplus R(\Lambda_{G_n\times 2D_{2m}}(\alpha))  \\ & \cong 
\prod_1^2 R(\Lambda_{G_n}(e^{\frac{2\pi i k}{n}}))\otimes_{\Z[q^{\pm}]} R(\Lambda_{2D_{2m}}(-I)) \\
&\cong \prod_1^2 R(\Lambda_{G_n}(e^{\frac{2\pi i k}{n}}))\otimes_{\Z[q^{\pm}]} \big( R(D_{2n}) [q^{\pm}]\oplus R_{[\widetilde{D_{2n}}_{\rho}]}(D_{2n}) [q^{\pm}] \big) \\ &\cong \prod_1^2 \big(R(D_{2n}) \oplus R_{[\widetilde{D_{2n}}_{\rho}]}(D_{2n}) \big) [x, q^{\pm}] / \langle x^{n}-q^{k} \rangle
\end{align*} where $\rho$ is the sign representation of $\Z/2$.

\end{itemize}

\item For the conjugacy class of $(e^{\frac{2\pi i k}{n}}, \tau)\in G_n\times 2D_{2m}$, 
The centralizer \[C_{G_n\times 2D_{2m}} (e^{\frac{2\pi i k}{n}}, \tau) = G_n \times \langle \tau \rangle \cong G_n\times \Z/4.\] \begin{align*} 
&&K_{\Lambda_{G_n\times 2D_{2m}} (e^{\frac{2\pi i k}{n}}, \tau)} (S^4)^{(e^{\frac{2\pi i k}{n}}, \tau)} \cong 
K_{\Lambda_{G_n\times 2D_{2m}} (e^{\frac{2\pi i k}{n}}, \tau)} (S^0) \\
&\cong &R(\Lambda_{G_n\times 2D_{2m}} (e^{\frac{2\pi i k}{n}}, \tau)) \oplus R(\Lambda_{G_n\times 2D_{2m}} (e^{\frac{2\pi i k}{n}}, \tau)) \\
&\cong &\Z[x_1, x_2, q^{\pm}]/\langle x^n_1-q^k, x_2^4-q \rangle \oplus  \Z[y_1, y_2, q^{\pm}]/\langle y^n_1-q^k, y_2^4-q \rangle.
\end{align*}

\item Then we study the case corresponding to the conjugacy class of $$(e^{\frac{2\pi i k}{n}}, \tau A_{\frac{2\pi i }{2m}})\in G_n\times 2D_{2m}.$$
The centralizer $C_{G_n\times 2D_{2m}} (e^{\frac{2\pi i k}{n}}, \tau A_{\frac{2\pi i }{2m}}) = G_n\times \langle \tau A_{\frac{2\pi i }{2m}} \rangle \cong G_n\times \Z/4.$
Thus, \begin{align*}
&&K_{\Lambda_{G_n\times 2D_{2m}} (e^{\frac{2\pi i k}{n}}, \tau A_{\frac{2\pi i }{2m}}) } (S^4)^{(e^{\frac{2\pi i k}{n}}, \tau A_{\frac{2\pi i }{2m}} )}
\buildrel{(\ast)}\over\cong K_{\Lambda_{G_n}(e^{\frac{2\pi i k}{n}}) \times_{\mathbb{T}} \Lambda_{2D_{2m}} (\tau A_{\frac{2\pi i }{2m}}) } (S^0) \\ &\cong & R(\Lambda_{G_n}(e^{\frac{2\pi i k}{n}}) \times_{\mathbb{T}} \Lambda_{2D_{2m}} (\tau A_{\frac{2\pi i }{2m}})) \oplus R(\Lambda_{G_n}(e^{\frac{2\pi i k}{n}}) \times_{\mathbb{T}} \Lambda_{2D_{2m}} (\tau A_{\frac{2\pi i }{2m}})) \\
&\cong &\Z[x_1, x_2, q^{\pm}]/\langle x_1^n- q^k, x_2^4-q \rangle \oplus \Z[y_1, y_2, q^{\pm}]/\langle y_1^n- q^k, y_2^4-q \rangle.
\end{align*} where the isomorphism $(\ast)$ is by \cite[Theorem 4.3]{atiyah1968_Bott}.

\end{enumerate}

\end{example}

\begin{example}

We compute  $\QR^{\ast}_{G_n\times 2D_{2m}} (S^4)$ in this example.
We take the Real structure $\hat{G}'_n$ and $2\hat{D}_{2m}$ as discussed in Example \ref{Real:G:SU(2)}. From them, we formulate 
a Real structure \[ \widehat{G_n\times 2D_{2m}} := \hat{G}'_n \times_{\Z/2}2\hat{D}_{2m}\] on the product $G'_n \times 2D_{2m}$.
By Example \ref{G_nDihedralS^4}, all the elements in $\pi_0(G_n\git_R \hat{G}'_n)$ are fixed points under the involution; and by Example \ref{2D2nS^4}, all the elements in $\pi_0(2D_{2m}\git_R 2\hat{D}_{2m} )$ are fixed points under the involution. Thus, all the points in $\pi_0 ( G_n\times 2D_{2m}\git_R \widehat{G_n\times 2D_{2m}})$ are fixed points. 

We  compute  the factors of $\QR^{\ast}_{G_n\times 2D_{2m}} (S^4)$ below one by one. We start with  those corresponding to the conjugacy classes 
represented by \[\alpha =( e^{\frac{2\pi i k}{n}},  e^{\frac{2\pi i p}{2m}}) \in G_n\times 2D_{2m}\] 
with $k$, $p\in \Z$.

\begin{enumerate}

\item 
 If $\alpha = I$, by \cite[Theorem 5.1]{atiyah1968_Bott}, \begin{align*}KR^{\ast}_{\Lambda_{G_n\times 2D_{2m}}(I)} (S^4)^{I} &= KR^{\ast}_{\Lambda_{G_n\times 2D_{2m}}(I )} (S^4)   \cong KR^{\ast}_{\Lambda_{G_n\times 2D_{2m}}(I)} (S^0) \\
 & \cong \prod_1^2 KR^{\ast}_{\Lambda_{G_n\times 2D_{2m}}(I)} (\pt) \cong \prod_1^2 KR^{\ast}_{\Lambda_{2D_{2m}}(I)}( \pt)\otimes_{KR^{\ast}_{\mathbb{T}}(\pt)}
 KR^{\ast}_{\Lambda_{G_n }(I)}(\pt) \\
 & \cong \prod_1^2 KR^{\ast}_{2D_{2m}}(\pt)[x, q^{\pm}]/\langle x^n -1\rangle.
\end{align*}

\item If the $e^{\frac{2\pi i p}{2m}}$ in  $\alpha $ is not $\pm I$, 
\begin{align*}
&&KR^{\ast}_{\Lambda_{G_n\times 2D_{2m}}(\alpha)} (S^4 )^{\alpha} \cong KR^{\ast}_{\Lambda_{G_n\times 2D_{2m}}(\alpha)} (S^0 )  \cong  KR^{\ast}_{\Lambda_{G_n\times 2D_{2m}}(\alpha)}(\pt )\oplus  KR^{\ast}_{\Lambda_{G_n\times 2D_{2m}}(\alpha)}(\pt )  \\ 
&\cong &KR^{\ast}(\pt)[x_1, x_2, q^{\pm}]/\langle x_1^n- q^{k}, x_2^{2m}-q^{p} \rangle \oplus KR^{\ast}(\pt)[x_1, x_2, q^{\pm}]/\langle x_1^n- q^{k}, x_2^{2m}-q^{p} \rangle.\end{align*}

\item If the $e^{\frac{2\pi i p}{2m}}$ in $\alpha$ is  $I$, 
\begin{align*}
&& KR^{\ast}_{\Lambda_{G_n\times 2D_{2m}}(\alpha)} (S^4 )^{\alpha} \cong KR^{\ast}_{\Lambda_{G_n\times 2D_{2m}}(\alpha)} (S^0 )
\cong  KR^{\ast}_{\Lambda_{G_n\times 2D_{2m}}(\alpha)} (\pt  ) \oplus  KR^{\ast}_{\Lambda_{G_n\times 2D_{2m}}(\alpha)} (\pt  )  \\ & 
\cong  & KR^{\ast}_{2D_{2m}}(\pt)[x, q^{\pm}]/\langle x^n- q^{k}  \rangle \oplus   KR^{\ast}_{2D_{2m}}(\pt)[x', q^{\pm}]/\langle x'^n- q^{k}  \rangle .\end{align*}

\item If the $e^{\frac{2\pi i p}{2m}}$ in $\alpha$ is  $-I$, applying Corollary \ref{RM:Z2}, we get \begin{align*}
& KR^{\ast}_{\Lambda_{G_n\times 2D_{2m}}(\alpha)} (S^4 )^{\alpha} \cong KR^{\ast}_{\Lambda_{G_n\times 2D_{2m}}(\alpha)} (S^0 ) \\  \cong  &KR^{\ast}_{\Lambda_{G_n\times 2D_{2m}}(\alpha)} (\pt ) \oplus KR^{\ast}_{\Lambda_{G_n\times 2D_{2m}}(\alpha)} (\pt)  \\  \cong & \prod_1^2
KR^{\ast}_{\Lambda_{G_n}(e^{\frac{2\pi i k}{n}})}(\pt)\otimes_{  KR^{\ast}_{\mathbb{T}} (\pt)} KR^{\ast}_{\Lambda_{2D_{2m}}(-1)}(\pt)  \\
\cong  & \prod_{1}^{2} KR^{\ast}_{\Lambda_{G_n}(e^{\frac{2\pi i k}{n}})}(\pt)\otimes_{  KR^{\ast}_{\mathbb{T}} (\pt)}  \bigg( KR^{\ast}_{2D_{2m}}(\pt) [q^{\pm}]\oplus KR^{\ast + \hat{\nu}_{\Lambda^R_{2\hat{D}_{2m}}(-I), sign}}_{2D_{2m}}(\pt)[q^{\pm}] \bigg) \\
\cong &  \prod_{1}^{2}\bigg( KR^{\ast}_{2D_{2m}}(\pt)\oplus KR^{\ast + \hat{\nu}_{\Lambda^R_{2\hat{D}_{2m}}(-I), sign}}_{2D_{2m}}(\pt) \bigg) [x, q^{\pm}]
/ \langle x^n-q^k\rangle 
\end{align*} where $sign$ is the sign representation of $\Z/2$.

\item 
For the conjugacy class of $(e^{\frac{2\pi i k}{n}}, \tau)\in G_n\times 2D_{2m}$, 
\begin{align*} 
& KR^{\ast}_{\Lambda_{G_n\times 2D_{2m}} (e^{\frac{2\pi i k}{n}}, \tau)} (S^4)^{(e^{\frac{2\pi i k}{n}}, \tau)} \cong 
KR^{\ast}_{\Lambda_{G_n\times 2D_{2m}} (e^{\frac{2\pi i k}{n}}, \tau)} (S^0) \\
\cong & \prod_1^2 KR^{\ast}_{\Lambda_{G_n\times 2D_{2m}} (e^{\frac{2\pi i k}{n}}, \tau)} (\pt) \\
\cong & \prod_1^2 KR^{\ast}(\pt)[x_1, x_2, q^{\pm}]/\langle x^n_1-q^k, x_2^4-q \rangle 
\end{align*}

\item For the conjugacy class of $(e^{\frac{2\pi i k}{n}}, \tau r)\in G_n\times 2D_{2m}$, 
\begin{align*}
& KR^{\ast}_{\Lambda_{G_n\times 2D_{2m}} (e^{\frac{2\pi i k}{n}}, \tau r) } (S^4)^{(e^{\frac{2\pi i k}{n}}, \tau r)}
\cong KR^{\ast}_{\Lambda_{G_n}(e^{\frac{2\pi i k}{n}}) \times_{\mathbb{T}} \Lambda_{2D_{2m}} (\tau r) } (S^0) \\ \cong  
& \prod_1^2 KR^{\ast}_{\Lambda_{G_n}(e^{\frac{2\pi i k}{n}}) \times_{\mathbb{T}} \Lambda_{2D_{2m}} (\tau r) } (\pt)  \\
\cong & \prod_1^2 KR^{\ast} (\pt)[x_1, x_2, q^{\pm}]/\langle x_1^n- q^k, x_2^4-q \rangle .
\end{align*}

\end{enumerate}
    
\end{example}

\begin{example}

In this example we deal with the finite subgroup $E_6\times E_7$ of $\mathrm{Spin}(5)$ where $E_6$ is the binary tetrahedral group and $E_7$ is the binary octahedral group, and compute the complex quasi-elliptic cohomology $QEll_{E_6\times E_7}(S^4)$.

First,  for the conjugacy classes $(\alpha, 1)$ where $\alpha$ is a conjugacy class in $E_6$ and $1$ represents the conjugacy classe consisting of itself in $E_7$, we have \begin{align*}
K_{\Lambda_{E_6\times E_7}(\alpha, 1)}(S^4)^{(\alpha, 1)} &\cong K_{\Lambda_{E_6}(\alpha)\times_{\mathbb{T}} \Lambda_{E_7}(1)} (S^4)^{\alpha}\\
&\cong 
K_{\Lambda_{E_6}(\alpha)}(S^0)\otimes RE_7 
\end{align*} Note that $\Lambda_{E_6}(\alpha)\times_{\mathbb{T}} \Lambda_{E_7}(1)\cong \Lambda_{E_6}(\alpha)\times_{\mathbb{T}} (C_{E_7}(1)\times \mathbb{T})
\cong \Lambda_{E_6}(\alpha)\times E_7$. The first factor $K_{\Lambda_{E_6}(\alpha)}(S^0)$ above is the factor in $QEll_{E_6}(S^4)$ corresponding to the conjugacy class $\alpha$, which is computed explicitly in \cite[Example 6.5]{huan2020_v2}.

And for the factors corresponding to the conjugacy classes $(1, \beta)$ where $\beta$ is a conjugacy class in $E_7$, as we discuss in Example \ref{factor_dual_spin5}, $(S^4)^{(1, \beta)}\cong (S^4)^{\beta}$, and \begin{align*}
K_{\Lambda_{E_6\times E_7}( 1, \beta )}(S^4)^{(1, \beta)} &\cong K_{\Lambda_{E_6}(1)\times_{\mathbb{T}} \Lambda_{E_7}(\beta)} (S^4)^{\beta}\\
&\cong 
K_{\Lambda_{E_7}(\beta)}(S^4)^{\beta}\otimes RE_6
\end{align*} where $K_{\Lambda_{E_7}(\beta)}(S^4)^{\beta}$ is the factor of $QEll_{E_7}(S^4)$ corresponding to the conjugacy class represented by $\beta$, which are all computed explicitly in \cite[Example 6.6]{huan2020_v2}.

Then, we think about the case corresponding to the conjugacy classes of the form $(\alpha, -1)$. By direct computation, \[(S^4)^{(\alpha, -1)} = (S^4)^{-\alpha}.\] We provide the conjugacy class of each $-\alpha$ and each fixed point space $(S^4)^{-\alpha}$ in Figure \ref{prod_ex_-1_E7}, where $a = \frac{1}{2} ( 1- i- j- k)$. 
\begin{figure} 
\begin{center}

\begin{tabular}{|c | c | c | c |c |} 
 \hline Representatives $\alpha$  &Centralizers &  & Conjugacy classes &   \\ 
 of Conjugacy classes & $C_{E_6}(\alpha)$ & $(S^4)^{\alpha}$ & of $-\alpha $ &$(S^4)^{- \alpha}$
 \\ \hline
$1$ & $ E_6$ & $S^4$  & $-1$ & $S^0$\\
$-1$ & $E_6$ & $ S^0$ & $1 $ & $S^4$\\
$i$ & $ \mathbb{Z}/4$ & $S^0$ & $ i $ & $S^0$\\
$a$ & $ \mathbb{Z}/6$ & $S^0$ & $-a $ & $S^0$\\
$-a$ & $ \mathbb{Z}/6$ & $S^0$ & $a $ & $S^0$\\
$a^2$ & $ \mathbb{Z}/6$  & $S^0$ & $-a^2$ & $S^0$\\
$-a^2$ & $ \mathbb{Z}/6$ & $S^0$ & $ a^2$ & $S^0$\\
 \hline
\end{tabular} \caption{Centralizers and fixed point spaces of $(\alpha, -1) \in E_6\times E_7$}
\end{center} \label{prod_ex_-1_E7}\end{figure}

In addition, we have the short exact sequence \begin{equation}
    1\longrightarrow \Z/2 \longrightarrow \Lambda_{E_6\times E_7}(\alpha, -1 )  \longrightarrow \Lambda_{E_6\times T_7}(\alpha, 1 ) \longrightarrow 1 
\end{equation} Note that the image of $\Z/2 \cong \{ (1, \pm 1) \}$ is contained in the center of $\Lambda_{E_6\times E_7}(\alpha, -1 )$, thus, we can apply Lemma \ref{dcld2}.

For $\alpha\neq -1$, the action of $\Z/2$ on $(S^4)^{-\alpha}$ is trivial. So we have
\begin{align*}
    K_{\Lambda_{E_6\times E_7}(\alpha, -1 ) } (S^4)^{-\alpha} &\cong K_{\Lambda_{E_6\times E_7}(\alpha, -1 ) } (\pt) \oplus K_{\Lambda_{E_6\times E_7}(\alpha, -1 ) } (\pt) \\ 
    & \cong \prod^2_1  R(\Lambda_{E_6}(\alpha)) \otimes  ( R (T_7) \oplus R_{[\widetilde{(T_7)}_{\rho}]} (T_7) ) 
\end{align*}  where $\rho$ is the sign representation of $\Z/2$.
Applying the computation in \cite[Example 6.5, Example 6.6]{huan2020_v2}, we  list the result of the computation of $ K_{\Lambda_{E_6\times E_7}(\alpha, -1 ) } (S^4)^{-\alpha}$ ($\alpha\neq -1$) below.
\begin{center}
    \begin{tabular}{|c|c|} \hline
    Representatives   $\alpha$   & The factor  \\
    of conjugacy classes  & $ K_{\Lambda_{E_6\times E_7}(\alpha, -1 ) } (S^4)^{-\alpha}$ \\
      \hline
         $1$    & $\prod\limits^2_1R(E_6)\otimes ( R (T_7) \oplus R_{[\widetilde{(T_7)}_{\rho}]} (T_7) ) [q^{\pm}]$  \\
$i$  & $\prod\limits^2_1( R (T_7) \oplus R_{[\widetilde{(T_7)}_{\rho}]} (T_7) ) [x, q^{\pm}]/\langle x^4-q\rangle $ \\
$a$   &$\prod\limits^2_1( R (T_7) \oplus R_{[\widetilde{(T_7)}_{\rho}]} (T_7) ) [x, q^{\pm}]/\langle x^6-q\rangle $  \\
$-a$    & $\prod\limits^2_1( R (T_7) \oplus R_{[\widetilde{(T_7)}_{\rho}]} (T_7) ) [x, q^{\pm}]/\langle x^6-q^4\rangle $ \\
$a^2$    & $\prod\limits^2_1( R (T_7) \oplus R_{[\widetilde{(T_7)}_{\rho}]} (T_7) ) [x, q^{\pm}]/\langle x^6-q^2\rangle $ \\
$-a^2$    & $\prod\limits^2_1( R (T_7) \oplus R_{[\widetilde{(T_7)}_{\rho}]} (T_7) ) [x, q^{\pm}]/\langle x^6-q^5\rangle $ \\
\hline
    \end{tabular}
\end{center}
Then we discuss the case that $\alpha= -1$. 

\begin{align*} K_{\Lambda_{E_6\times E_7} (-1, -1)} (S^4)
&\cong K_{\Lambda_{E_6\times E_7} (-1, -1)} (S^0) \cong K_{\Lambda_{E_6\times E_7} (-1, -1)} (\pt) \oplus K_{\Lambda_{E_6\times E_7} (-1, -1)} (\pt) \\ & \cong \prod\limits_1^2 K_{\Lambda_{E_6} (-1)} (\pt) \otimes_{\Z[q^{\pm}]} K_{\Lambda_{ E_7} ( -1)} (\pt) \\ &\cong   \prod\limits_1^2 \bigg(R (T_6) \oplus R_{[\widetilde{(T_6)}_{\rho}]} (T_6) \bigg) \otimes \bigg( R (T_7) \oplus R_{[\widetilde{(T_7)}_{\rho}]} (T_7)  \bigg) [q^{\pm}]  
\end{align*} where  $\rho$ is the sign representation of $\Z/2$.

\bigskip

Next, we deal with the factor corresponding to the conjugacy classes \[( \alpha, i ) \]
By direct computation, we can get 
\begin{center}
    \begin{tabular}{|c|c |c|}
     \hline Representatives $\alpha$   & The fixed point space   & $\Lambda_{E_6\times E_7} ( \alpha, i) $\\ 
 of Conjugacy classes   & $(S^4)^{(\alpha, i)}$  & is isomorphic to
 \\ \hline
         $1$ &  $S^0$ & $E_6\times \Lambda_{\Z/8}(2) $\\
         $-1$ &  $S^0$ & $\Lambda_{E_6}(-1)\times_{\mathbb{T}} \Lambda_{\Z/8}(2) $ \\
$i$ & $\{(a, b, 0, 0)\in\R^4\}\cup\{\infty\}\cong S^2$ & $\Lambda_{\Z/4}(1)\times_{\mathbb{T}} \Lambda_{\Z/8}(2) $   \\
$a$ & $S^0$ & $\Lambda_{\Z/6}(1)\times_{\mathbb{T}} \Lambda_{\Z/8}(2) $\\
$-a $ &  $S^0$  & $\Lambda_{\Z/6}(4)\times_{\mathbb{T}} \Lambda_{\Z/8}(2) $ \\
$a^2$ & $S^0$ &  $\Lambda_{\Z/6}(2)\times_{\mathbb{T}} \Lambda_{\Z/8}(2) $ \\
$-a^2$ & $S^0$  & $\Lambda_{\Z/6}(5)\times_{\mathbb{T}} \Lambda_{\Z/8}(2) $ \\
\hline 
    \end{tabular}
\end{center}

For $\alpha = -1$, \begin{align*}
    K_{\Lambda_{E_6\times E_7}(\alpha, i)} (S^4)^{(\alpha, i)} &\cong 
    K_{\Lambda_{E_6}(-1)\times_{\mathbb{T}} \Lambda_{\Z/8}(2)} (S^0)\\
&    \cong \prod\limits^2_1  R(\Lambda_{E_6}(-1) )\otimes_{\Z[q^{\pm}]} R (\Lambda_{\Z/8}(2)) \\
& \cong \prod\limits^2_1 \bigg(R (T_6) \oplus R_{[\widetilde{(T_6)}_{\rho}]} (T_6) \bigg) [ x, q^{\pm}]/\langle x^8-q^2\rangle
\end{align*}  where  $\rho$ is the sign representation of $\Z/2$.

We list the computation of the other cases $K_{\Lambda_{E_6\times E_7}(\alpha, i)} (S^4)^{(\alpha, i)}$ ($\alpha \neq -1$) below. 
\begin{center}
    \begin{tabular}{|c|c |c|}
     \hline Representatives $\alpha$   & The factor   \\ 
 of Conjugacy classes   & $K_{\Lambda_{E_6\times E_7}(\alpha, i)} (S^4)^{(\alpha, i)}$  
 \\ \hline
         $1$ &  $\prod\limits_1^2 R(E_6)[x, q^{\pm}]/\langle x^8-q^2\rangle$\\
$i$ &  $\prod\limits_1^2 \Z[x_1, x_2, q^{\pm}]/\langle x_1^4-q, x_2^8-q^2 \rangle $ \\
$a$ & $ \prod\limits_1^2 \Z[x_1, x_2, q^{\pm}]/\langle x_1^6-q, x_2^8-q^2 \rangle $\\
$-a $ &  $\prod\limits_1^2 \Z[x_1, x_2, q^{\pm}]/\langle x_1^6-q^4, x_2^8-q^2 \rangle $ \\
$a^2$ & $\prod\limits_1^2 \Z[x_1, x_2, q^{\pm}]/\langle x_1^6-q^2, x_2^8-q^2 \rangle $ \\
$-a^2$ & $\prod\limits_1^2 \Z[x_1, x_2, q^{\pm}]/\langle x_1^6-q^5, x_2^8-q^2 \rangle $ \\
\hline 
    \end{tabular}
\end{center}

Next we deal with the conjugacy classes \[(\alpha, s= \frac{1}{2}(1+i+j+k) ),\]
and compute the factors $K_{\Lambda_{E_6\times E_7}(\alpha, i)} (S^4)^{(\alpha, s)}$.

By direct computation, we can get 
\begin{center}
    \begin{tabular}{|c|c |c|}\hline 
      $\alpha $   & $(S^4)^{(\alpha, s)}$ & $\Lambda_{E_6\times E_7} ( \alpha, s) $ \\
           \hline
         $1$ &  $S^0$ & $E_6\times \Lambda_{\Z/6}(1) $\\
         $-1$ &  $S^0$ & $\Lambda_{E_6}(-1)\times_{\mathbb{T}} \Lambda_{\Z/6}(1) $ \\
$i$ & $S^0$ & $\Lambda_{\Z/4}(1)\times_{\mathbb{T}} \Lambda_{\Z/6}(1) $   \\
$a$ & $\{(0,-c-d, c, d)\in \R^4\}\cap \{\infty\} \cong S^2$ & $\Lambda_{\Z/6}(1)\times_{\mathbb{T}} \Lambda_{\Z/6}(1) $\\
$-a$ &  $S^0$  & $\Lambda_{\Z/6}(4)\times_{\mathbb{T}} \Lambda_{\Z/6}(1) $ \\
$a^2$ & $S^0$ &  $\Lambda_{\Z/6}(2)\times_{\mathbb{T}} \Lambda_{\Z/6}(1) $ \\
$-a^2$ & $S^0$  & $\Lambda_{\Z/6}(5)\times_{\mathbb{T}} \Lambda_{\Z/6}(1) $ \\
\hline 
    \end{tabular}
\end{center}

We list the computation of $K_{\Lambda_{E_6\times E_7}(\alpha, s)} (S^4)^{(\alpha, s)}$  below.

\begin{center}
    \begin{tabular}{|c|c |c|}
     \hline Representatives $\alpha$   & The factor   \\ 
 of Conjugacy classes   & $K_{\Lambda_{E_6\times E_7}(\alpha, s)} (S^4)^{(\alpha, s)}$ 
 \\ \hline
         $1$ &  $\prod\limits_1^2 R(E_6)[x, q^{\pm}]/\langle x^6-q\rangle$\\
                  $-1$ &  $\prod\limits_1^2 \bigg(R (T_6) \oplus R_{[\widetilde{(T_6)}_{\rho}]} (T_6) \bigg) [x, q^{\pm}]/\langle x^6-q\rangle$\\
$i$ &  $\prod\limits_1^2 \Z[x_1, x_2, q^{\pm}]/\langle x_1^4-q, x_2^6-q \rangle $ \\
$a$ & $ \prod\limits_1^2 \Z[x_1, x_2, q^{\pm}]/\langle x_1^6-q, x_2^6-q \rangle $\\
$-a $ &  $\prod\limits_1^2 \Z[x_1, x_2, q^{\pm}]/\langle x_1^6-q^4, x_2^6-q \rangle $ \\
$a^2$ & $\prod\limits_1^2 \Z[x_1, x_2, q^{\pm}]/\langle x_1^6-q^2, x_2^6-q \rangle $ \\
$-a^2$ & $\prod\limits_1^2 \Z[x_1, x_2, q^{\pm}]/\langle x_1^6-q^5, x_2^6-q \rangle $ \\
\hline 
    \end{tabular}
\end{center}

Next we deal with the conjugacy classes \[(\alpha, -s= -\frac{1}{2}(1+i+j+k) ).\]

By direct computation, we can get 
\begin{center}
    \begin{tabular}{|c|c |c|}\hline 
      $\alpha $   & $(S^4)^{(\alpha, -s)}$ & $\Lambda_{E_6\times E_7} ( \alpha, -s) $ \\
           \hline
         $1$ &  $S^0$ & $E_6\times \Lambda_{\Z/6}(4) $\\
         $-1$ &  $S^0$ & $\Lambda_{E_6}(-1)\times_{\mathbb{T}} \Lambda_{\Z/6}(4) $ \\
$i$ & $S^0$ & $\Lambda_{\Z/4}(1)\times_{\mathbb{T}} \Lambda_{\Z/6}(4) $   \\
$a$ & $S^0$ & $\Lambda_{\Z/6}(1)\times_{\mathbb{T}} \Lambda_{\Z/6}(4) $\\
$-a$ &  $S^0$  & $\Lambda_{\Z/6}(4)\times_{\mathbb{T}} \Lambda_{\Z/6}(4) $ \\
$a^2$ & $S^0$ &  $\Lambda_{\Z/6}(2)\times_{\mathbb{T}} \Lambda_{\Z/6}(4) $ \\
$-a^2$ & $S^0$  & $\Lambda_{\Z/6}(5)\times_{\mathbb{T}} \Lambda_{\Z/6}(4) $ \\
\hline 
    \end{tabular}
\end{center}

We list the computation of $K_{\Lambda_{E_6\times E_7}(\alpha, -s)} (S^4)^{(\alpha, -s)}$  below.

\begin{center}
    \begin{tabular}{|c|c |c|}
     \hline Representatives $\alpha$   & The factor   \\ 
 of Conjugacy classes   & $K_{\Lambda_{E_6\times E_7}(\alpha, -s)} (S^4)^{(\alpha, -s)}$ 
 \\ \hline
         $1$ &  $\prod\limits_1^2 R(E_6)[x, q^{\pm}]/\langle x^6-q^4\rangle$\\
                  $-1$ &  $\prod\limits_1^2 \bigg(R (T_6) \oplus R_{[\widetilde{(T_6)}_{\rho}]} (T_6) \bigg) [x, q^{\pm}]/\langle x^6-q^4\rangle$\\
$i$ &  $\prod\limits_1^2 \Z[x_1, x_2, q^{\pm}]/\langle x_1^4-q, x_2^6-q^4 \rangle $ \\
$a$ & $ \prod\limits_1^2 \Z[x_1, x_2, q^{\pm}]/\langle x_1^6-q, x_2^6-q^4 \rangle $\\
$-a $ &  $\prod\limits_1^2 \Z[x_1, x_2, q^{\pm}]/\langle x_1^6-q^4, x_2^6-q^4 \rangle $ \\
$a^2$ & $\prod\limits_1^2 \Z[x_1, x_2, q^{\pm}]/\langle x_1^6-q^2, x_2^6-q^4 \rangle $ \\
$-a^2$ & $\prod\limits_1^2 \Z[x_1, x_2, q^{\pm}]/\langle x_1^6-q^5, x_2^6-q^4 \rangle $ \\
\hline 
    \end{tabular}
\end{center}

Next we deal with the conjugacy classes \[(\alpha, r=\frac{1}{\sqrt{2}} (i+j) ),\] and compute the factors
$K_{\Lambda_{E_6\times E_7}(\alpha, r)} (S^4)^{(\alpha, r)}$.

By direct computation, we can get 
\begin{center}
    \begin{tabular}{|c|c |c|}\hline 
      $\alpha $   & $(S^4)^{(\alpha, r)}$ & $\Lambda_{E_6\times E_7} ( \alpha, r) $ \\
           \hline
         $1$ &  $S^0$ & $E_6\times \Lambda_{\Z/4}(1) $\\
         $-1$ &  $S^0$ & $\Lambda_{E_6}(-1)\times_{\mathbb{T}} \Lambda_{\Z/4}(1) $ \\
$i$ & $S^0$ & $\Lambda_{\Z/4}(1)\times_{\mathbb{T}} \Lambda_{\Z/4}(1) $   \\
$a$ & $S^0$ & $\Lambda_{\Z/6}(1)\times_{\mathbb{T}} \Lambda_{\Z/4}(1) $\\
$-a$ &  $S^0$  & $\Lambda_{\Z/6}(4)\times_{\mathbb{T}} \Lambda_{\Z/4}(1) $ \\
$a^2$ & $S^0$ &  $\Lambda_{\Z/6}(2)\times_{\mathbb{T}} \Lambda_{\Z/4}(1) $ \\
$-a^2$ & $S^0$  & $\Lambda_{\Z/6}(5)\times_{\mathbb{T}} \Lambda_{\Z/4}(1) $ \\
\hline 
    \end{tabular}
\end{center}

We list the computation of $K_{\Lambda_{E_6\times E_7}(\alpha, r)} (S^4)^{(\alpha, r)}$  below.

\begin{center}
    \begin{tabular}{|c|c |c|}
     \hline Representatives $\alpha$   & The factor   \\ 
 of Conjugacy classes   & $K_{\Lambda_{E_6\times E_7}(\alpha, r)} (S^4)^{(\alpha, r)}$ 
 \\ \hline
         $1$ &  $\prod\limits_1^2 R(E_6)[x, q^{\pm}]/\langle x^4-q\rangle$\\
                  $-1$ &  $\prod\limits_1^2 \bigg(R (T_6) \oplus R_{[\widetilde{(T_6)}_{\rho}]} (T_6) \bigg) [x, q^{\pm}]/\langle x^4-q\rangle$\\
$i$ &  $\prod\limits_1^2 \Z[x_1, x_2, q^{\pm}]/\langle x_1^4-q, x_2^4-q \rangle $ \\
$a$ & $ \prod\limits_1^2 \Z[x_1, x_2, q^{\pm}]/\langle x_1^6-q, x_2^4-q \rangle $\\
$-a $ &  $\prod\limits_1^2 \Z[x_1, x_2, q^{\pm}]/\langle x_1^6-q^4, x_2^4-q \rangle $ \\
$a^2$ & $\prod\limits_1^2 \Z[x_1, x_2, q^{\pm}]/\langle x_1^6-q^2, x_2^4-q \rangle $ \\
$-a^2$ & $\prod\limits_1^2 \Z[x_1, x_2, q^{\pm}]/\langle x_1^6-q^5, x_2^4-q \rangle $ \\
\hline 
    \end{tabular}
\end{center}

Next we deal with the conjugacy classes \[(\alpha, t= \frac{1}{\sqrt{2}}(1+i) ),\] and compute the factors
$K_{\Lambda_{E_6\times E_7}(\alpha, t)} (S^4)^{(\alpha, t)}$.

By direct computation, we can get 
\begin{center}
    \begin{tabular}{|c|c |c|}\hline 
      $\alpha $   & $(S^4)^{(\alpha, t)}$ & $\Lambda_{E_6\times E_7} ( \alpha, t) $ \\
           \hline
         $1$ &  $S^0$ & $E_6\times \Lambda_{\Z/8}(1) $\\
         $-1$ &  $S^0$ & $\Lambda_{E_6}(-1)\times_{\mathbb{T}} \Lambda_{\Z/8}(1) $ \\
$i$ & $S^0$ & $\Lambda_{\Z/4}(1)\times_{\mathbb{T}} \Lambda_{\Z/8}(1) $   \\
$a$ & $S^0$ & $\Lambda_{\Z/6}(1)\times_{\mathbb{T}} \Lambda_{\Z/8}(1) $\\
$-a$ &  $S^0$  & $\Lambda_{\Z/6}(4)\times_{\mathbb{T}} \Lambda_{\Z/8}(1) $ \\
$a^2$ & $S^0$ &  $\Lambda_{\Z/6}(2)\times_{\mathbb{T}} \Lambda_{\Z/8}(1) $ \\
$-a^2$ & $S^0$  & $\Lambda_{\Z/6}(5)\times_{\mathbb{T}} \Lambda_{\Z/8}(1) $ \\
\hline 
    \end{tabular}
\end{center}

We list the computation of $K_{\Lambda_{E_6\times E_7}(\alpha, t)} (S^4)^{(\alpha, t)}$  below.

\begin{center}
    \begin{tabular}{|c|c |c|}
     \hline Representatives $\alpha$   & The factor   \\ 
 of Conjugacy classes   & $K_{\Lambda_{E_6\times E_7}(\alpha, t)} (S^4)^{(\alpha, t)}$ 
 \\ \hline
         $1$ &  $\prod\limits_1^2 R(E_6)[x, q^{\pm}]/\langle x^8-q\rangle$\\
                  $-1$ &  $\prod\limits_1^2 \bigg(R (T_6) \oplus R_{[\widetilde{(T_6)}_{\rho}]} (T_6) \bigg) [x, q^{\pm}]/\langle x^8-q\rangle$\\
$i$ &  $\prod\limits_1^2 \Z[x_1, x_2, q^{\pm}]/\langle x_1^4-q, x_2^8-q \rangle $ \\
$a$ & $ \prod\limits_1^2 \Z[x_1, x_2, q^{\pm}]/\langle x_1^6-q, x_2^8-q \rangle $\\
$-a $ &  $\prod\limits_1^2 \Z[x_1, x_2, q^{\pm}]/\langle x_1^6-q^4, x_2^8-q \rangle $ \\
$a^2$ & $\prod\limits_1^2 \Z[x_1, x_2, q^{\pm}]/\langle x_1^6-q^2, x_2^8-q \rangle $ \\
$-a^2$ & $\prod\limits_1^2 \Z[x_1, x_2, q^{\pm}]/\langle x_1^6-q^5, x_2^8-q \rangle $ \\
\hline 
    \end{tabular}
\end{center}

Next we deal with the conjugacy classes \[(\alpha, -t= -\frac{1}{\sqrt{2}}(1+i) ),\] and compute the factors
$K_{\Lambda_{E_6\times E_7}(\alpha, -t)} (S^4)^{(\alpha, -t)}$.

By direct computation, we can get

\begin{center}
    \begin{tabular}{|c|c |c|}\hline 
      $\alpha $   & $(S^4)^{(\alpha, -t)}$ & $\Lambda_{E_6\times E_7} ( \alpha, -t) $ \\
           \hline
         $1$ &  $S^0$ & $E_6\times \Lambda_{\Z/8}(5) $\\
         $-1$ &  $S^0$ & $\Lambda_{E_6}(-1)\times_{\mathbb{T}} \Lambda_{\Z/8}(5) $ \\
$i$ & $S^0$ & $\Lambda_{\Z/4}(1)\times_{\mathbb{T}} \Lambda_{\Z/8}(5) $   \\
$a$ & $S^0$ & $\Lambda_{\Z/6}(1)\times_{\mathbb{T}} \Lambda_{\Z/8}(5) $\\
$-a$ &  $S^0$  & $\Lambda_{\Z/6}(4)\times_{\mathbb{T}} \Lambda_{\Z/8}(5) $ \\
$a^2$ & $S^0$ &  $\Lambda_{\Z/6}(2)\times_{\mathbb{T}} \Lambda_{\Z/8}(5) $ \\
$-a^2$ & $S^0$  & $\Lambda_{\Z/6}(5)\times_{\mathbb{T}} \Lambda_{\Z/8}(5) $ \\
\hline 
    \end{tabular}
\end{center}

We list the computation of $K_{\Lambda_{E_6\times E_7}(\alpha, -t)} (S^4)^{(\alpha, -t)}$  below.

\begin{center}
    \begin{tabular}{|c|c |c|}
     \hline Representatives $\alpha$   & The factor   \\ 
 of Conjugacy classes   & $K_{\Lambda_{E_6\times E_7}(\alpha, -t)} (S^4)^{(\alpha, -t)}$ 
 \\ \hline
         $1$ &  $\prod\limits_1^2 R(E_6)[x, q^{\pm}]/\langle x^8-q^5\rangle$\\
                  $-1$ &  $\prod\limits_1^2 \bigg(R (T_6) \oplus R_{[\widetilde{(T_6)}_{\rho}]} (T_6) \bigg) [x, q^{\pm}]/\langle x^8-q^5\rangle$\\
$i$ &  $\prod\limits_1^2 \Z[x_1, x_2, q^{\pm}]/\langle x_1^4-q, x_2^8-q^5 \rangle $ \\
$a$ & $ \prod\limits_1^2 \Z[x_1, x_2, q^{\pm}]/\langle x_1^6-q, x_2^8-q^5 \rangle $\\
$-a $ &  $\prod\limits_1^2 \Z[x_1, x_2, q^{\pm}]/\langle x_1^6-q^4, x_2^8-q^5 \rangle $ \\
$a^2$ & $\prod\limits_1^2 \Z[x_1, x_2, q^{\pm}]/\langle x_1^6-q^2, x_2^8-q^5 \rangle $ \\
$-a^2$ & $\prod\limits_1^2 \Z[x_1, x_2, q^{\pm}]/\langle x_1^6-q^5, x_2^8-q^5 \rangle $ \\
\hline 
    \end{tabular}
\end{center}

\end{example}

\begin{example}
    In this example we  compute the Real quasi-elliptic cohomology $\QR^{\ast}_{E_6\times E_7}(S^4)$. We take the Real structure of $E_6$ given in Example \ref{E6_S4_Real} and the Real structure of $E_7$ given in Example \ref{E7_S4_Real}.
    Note that, an element $(h, k) \in \widehat{E_6\times E_7}  $ is a fixed point under the reflection if and only if $h$  is a fixed point in $E_6$ and $k$ is a fixed point in $E_7$; in addition, an element $(h, k) \in \widehat{E_6\times E_7}  $ is a free point under the reflection if and only if $h$  is a free point in $E_6$ and $k$ is a free  point in $E_7$. 

    As shown in Example \ref{E7_S4_Real}, all the representatives of the conjugacy classes in $E_7$, as given in Figure \ref{E7:conj:c:fps}, are fixed points under the reflection. Thus, all the representatives of the conjugacy classes in $E_6\times E_7$ are fixed points and they are represented by the elements $(h, k ) \in E_6\times E_7$ with $h$ a fixed point. Then, by Figure \ref{E6:conj}, $h$ can only be $1$, $-1$ and $j$.

We first deal with the conjugacy classes \[(1, \beta ),\] where $\beta$ goes over all the representatives of the conjugacy classes in $E_7$ and compute the factors
$KR^{\ast}_{\Lambda_{E_6\times E_7}(1, \beta)} (S^4)^{(1, \beta)}$.

Applying Proposition \ref{prod_QEll}, we list the computation of $KR^{\ast}_{\Lambda_{E_6\times E_7}(1, \beta)} (S^4)^{(1, \beta)}$  below.

\begin{center}
    \begin{tabular}{|c|c |c|}
     \hline Representatives $\beta$   & The factor   \\ 
 of Conjugacy classes   & $KR^{\ast}_{\Lambda_{E_6\times E_7}(1, \beta)} (S^4)^{(1, \beta)}$ 
 \\ \hline
         $1$ &  $\prod\limits_1^2  KR^{\ast}_{E_6}(\pt) \otimes_{KR^{\ast}(\pt)} KR^{\ast}_{E_7}(\pt) [q^{\pm}] $\\
        $-1$ &  $\prod\limits_1^2  KR^{\ast}_{E_6}(\pt)\otimes_{KR^{\ast}(\pt)} \bigg( KR^{\ast}_{T_7}(\pt)  \oplus KR^{\ast + \hat{\nu}_{\Lambda^R_{\hat{E_7'}(-I)}, sign}}_{T_7}(\pt) \bigg) [q^{\pm}] $ \\
$j $ & $ \prod\limits_1^2   KR^{\ast}_{E_6}(\pt)[x, q^{\pm }] /\langle x^8-q^2\rangle  $ \\
$\theta $ &  $\prod\limits_1^2   KR^{\ast}_{E_6}(\pt)[x, q^{\pm }] /\langle x^6-q\rangle    $ \\
$- \theta$ & $\prod\limits_1^2   KR^{\ast}_{E_6}(\pt)[x, q^{\pm }] /\langle x^6-q^4\rangle   $ \\
$r$ & $\prod\limits_1^2   KR^{\ast}_{E_6}(\pt)[x, q^{\pm }] /\langle x^4-q\rangle    $ \\
$t$ &  $\prod\limits_1^2   KR^{\ast}_{E_6}(\pt)[x, q^{\pm }] /\langle x^8-q\rangle   $ \\
$-t$ &  $\prod\limits_1^2   KR^{\ast}_{E_6}(\pt)[x, q^{\pm }] /\langle x^8-q^5\rangle $   \\
\hline 
    \end{tabular}
\end{center}

Next, 
applying Proposition \ref{prod_QEll}, we list the computation of $KR^{\ast}_{\Lambda_{E_6\times E_7}(-1, \beta)} (S^4)^{(-1, \beta)}$  below.

\begin{center}
    \begin{tabular}{|c|c |c|}
     \hline Representatives $\beta$   & The factor   \\ 
 of Conjugacy classes   & $KR^{\ast}_{\Lambda_{E_6\times E_7}(-1, \beta)} (S^4)^{(-1, \beta)}$ 
 \\ \hline
         $1$ &  $\prod\limits_1^2  \bigg(  KR^{\ast}_{T_{6}}(\pt) \oplus KR^{\ast + \hat{\nu}_{\Lambda^R_{\hat{E_6'}(-I)}, sign}}_{T_{6}}(\pt)  \bigg)\otimes_{KR^{\ast}(\pt)} KR^{\ast}_{E_7}(\pt) [q^{\pm}] $\\
        $-1$ &  $\prod\limits_1^2  \bigg(  KR^{\ast}_{T_{6}}(\pt) \oplus KR^{\ast + \hat{\nu}_{\Lambda^R_{\hat{E_6'}(-I)}, sign}}_{T_{6}}(\pt)  \bigg)$ \\ & $\otimes_{KR^{\ast}(\pt)} \bigg( KR^{\ast}_{T_7}(\pt)  \oplus KR^{\ast + \hat{\nu}_{\Lambda^R_{\hat{E_7'}(-I)}, sign}}_{T_7}(\pt) \bigg) [q^{\pm}] $ \\
$j $ & $ \prod\limits_1^2   \bigg(  KR^{\ast}_{T_{6}}(\pt) \oplus KR^{\ast + \hat{\nu}_{\Lambda^R_{\hat{E_6'}(-I)}, sign}}_{T_{6}}(\pt)  \bigg)[x, q^{\pm }] /\langle x^8-q^2\rangle  $ \\
$\theta $ &  $\prod\limits_1^2   \bigg(  KR^{\ast}_{T_{6}}(\pt) \oplus KR^{\ast + \hat{\nu}_{\Lambda^R_{\hat{E_6'}(-I)}, sign}}_{T_{6}}(\pt)  \bigg)[x, q^{\pm }] /\langle x^6-q\rangle    $ \\
$- \theta$ & $\prod\limits_1^2   \bigg(  KR^{\ast}_{T_{6}}(\pt) \oplus KR^{\ast + \hat{\nu}_{\Lambda^R_{\hat{E_6'}(-I)}, sign}}_{T_{6}}(\pt)  \bigg)[x, q^{\pm }] /\langle x^6-q^4\rangle   $ \\
$r$ & $\prod\limits_1^2   \bigg(  KR^{\ast}_{T_{6}}(\pt) \oplus KR^{\ast + \hat{\nu}_{\Lambda^R_{\hat{E_6'}(-I)}, sign}}_{T_{6}}(\pt)  \bigg)[x, q^{\pm }] /\langle x^4-q\rangle    $ \\
$t$ &  $\prod\limits_1^2   \bigg(  KR^{\ast}_{T_{6}}(\pt) \oplus KR^{\ast + \hat{\nu}_{\Lambda^R_{\hat{E_6'}(-I)}, sign}}_{T_{6}}(\pt)  \bigg)[x, q^{\pm }] /\langle x^8-q\rangle   $ \\
$-t$ &  $\prod\limits_1^2   \bigg(  KR^{\ast}_{T_{6}}(\pt) \oplus KR^{\ast + \hat{\nu}_{\Lambda^R_{\hat{E_6'}(-I)}, sign}}_{T_{6}}(\pt)  \bigg)[x, q^{\pm }] /\langle x^8-q^5\rangle $   \\
\hline 
    \end{tabular}
    
\end{center}

    In addition, we list computation of $KR^{\ast}_{\Lambda_{E_6\times E_7}(j, \beta)} (S^4)^{(j, \beta)}$ below.

\begin{center}
    \begin{tabular}{|c|c |c|}
     \hline Representatives $\beta$   & The factor   \\ 
 of Conjugacy classes   & $KR^{\ast}_{\Lambda_{E_6\times E_7}(j, \beta)} (S^4)^{(j, \beta)}$ 
 \\ \hline
         $1$ &  $\prod\limits_1^2  KR^{\ast}_{E_7}(\pt) [y, q^{\pm}] /\langle y^4- q\rangle$\\
        $-1$ &  $\prod\limits_1^2   \bigg( KR^{\ast}_{T_7}(\pt)  \oplus KR^{\ast + \hat{\nu}_{\Lambda^R_{\hat{E_7'}(-I)}, sign}}_{T_7}(\pt) \bigg) [y, q^{\pm}] /\langle y^4- q\rangle$ \\
$j $ & $ \prod\limits_1^2   KR^{\ast}(\pt)[x, y, q^{\pm }] /\langle y^4-q, x^8-q^2\rangle  $ \\
$\theta $ &  $\prod\limits_1^2   KR^{\ast}(\pt)[x, y, q^{\pm }] /\langle y^4-q,  x^6-q\rangle    $ \\
$- \theta$ & $\prod\limits_1^2   KR^{\ast}(\pt)[x, y, q^{\pm }]  /\langle y^4-q, x^6-q^4\rangle   $ \\
$r$ & $\prod\limits_1^2   KR^{\ast}(\pt)[x, y, q^{\pm }]  /\langle y^4-q,  x^4-q\rangle    $ \\
$t$ &  $\prod\limits_1^2   KR^{\ast}(\pt)[x, y, q^{\pm }] /\langle y^4-q,  x^8-q\rangle   $ \\
$-t$ &  $\prod\limits_1^2   KR^{\ast}(\pt)[x, y, q^{\pm }] /\langle y^4-q,  x^8-q^5\rangle $   \\
\hline 
    \end{tabular}
\end{center}

\end{example}

\newpage

\appendix

\section{Corollaries of {\'A}ngel-G{\'o}mez-Uribe Decomposition Formula} \label{Cor:decomp:dcl}

In this section, we prove some corollaries of \cite[Theorem 3.6, Corollary 3.7]{ngel2017EquivariantCB}.
They all apply to compact Lie groups.

\begin{lemma} 
Let $Q$ and $G$ be compact Lie groups. And we have a short exact sequence \[ 1 \longrightarrow \Z/2  \buildrel{l}\over\longrightarrow G \buildrel{\pi}\over\longrightarrow Q \longrightarrow 1 \] and $l(A)$ is contained in the center of $G$. Let $X$ be a $G$-space with $l(\Z/2)$ acting on it trivially.
Then, we have the isomorphism
\[K^*_{G}(X)\cong K^*_{Q}(X)\oplus K^{[\tilde{Q}_{sign}]+*}_{Q}(X)\]
\label{dcl}
\end{lemma}
\begin{proof}
As given in \cite[Section 2.1]{ngel2017EquivariantCB}, there is a well-defined $G$-action on the irreducible $\Z/2$-representations by \[(g\cdot \rho)(a) = \rho (g^{-1} a g) = \rho (a),\] for any $g\in G$, $a\in \Z/2$ and any irreducible $\Z/2$-representation $\rho$.  

Since the irreducible representations $(\rho, V_{\rho})$ of $\mathbb{Z}/2$ are all 1-dimensional and fixed by $G$,  
 the group $PU(1)$ of inner automorphism of $U(1)$ consists of exactly one element, i.e. the identity map. As in \cite[(1), page 6]{ngel2017EquivariantCB}, we use the symbol $\tilde{G}_{\rho}$ to denote the pullback  
\[ \xymatrix{\tilde{G}_{\rho} \ar[r]^{\tilde{f}} \ar[d]^{\tau_{\rho}} &U(1) \ar[d] \\
G\ar[r] &PU(1)} \] We have $\tilde{G}_{\rho} = G\times U(1)$. The map $\tau_{\rho}$ is the projection map to $G$ and $\tilde{f}$ is the projection map to $U(1)$. 

Then we consider the commutative diagram
\[ \xymatrix{\Z/2 \ar[r]^{\tilde{l}} \ar[d]_{=} &\tilde{G}_{\rho} \ar[d] \\ \Z/2 \ar[r]^{l} &G}\] where $\tilde{l}$ is defined to be the unique map so that $\rho=\tilde{f}\circ \tilde{l}$. Thus, 
$\tilde{l}$ is the product of $l$ and the representation $\rho$. 

Then
we consider the commutative diagram \begin{equation} \label{qrho:plbk}
    \xymatrix{  &\Z/2 \ar[d]^{\tilde{l}} & \Z/2\ar[d]^{l} \\ \T \ar[r] & \tilde{G}_{\rho} \ar[d]^{\tilde{\pi}} \ar[r] &G \ar[d]^{\pi} \\
\T\ar[r]^{i_Q} & \tilde{Q}_{\rho} \ar[r]^{p_Q} & Q   }
\end{equation} where the vertical sequences are both exact, the horizontal sequences are $\T$-central extensions and the square is a pullback square. 
If $\rho$ is the trivial representation of $\Z/2$, $\tilde{Q}_{\rho} \cong Q\times \T$ and, by \cite[Proposition 2.2]{ngel2017EquivariantCB}, $\rho$ extends to an irreducible representation of $G$. However, if $\rho$ is the sign representation of $\Z/2$,  it may not extend to the whole group $G$. And the central extension \[ \xymatrix{1\ar[r] & \T \ar[r]^{i_Q} & \tilde{Q}_{\rho}  \ar[r]^{p_Q}
& Q \ar[r] &1}\] may correspond to a nontrivial element $[\tilde{Q}_{\rho}]$  in $H^3(BQ; \mathbb{Z})$.




By \cite[Corollary 3.7]{ngel2017EquivariantCB},
\begin{equation}
K^*_{G}(X) \cong \bigoplus_{\rho\in G \slash Irr(\Z/2)} K^{[\tilde{Q}_{\rho}] +\ast}_{Q_{\rho}} (X), \label{decomp:z2:cor}\end{equation}
where $\rho$ runs over representatives of the orbits of the $G$-action on the set of isomorphism classes of irreducible $\Z/2$-representations, i.e. $\{1, sign\}$, the action of \[Q_{\rho} = G_{\rho}/(\Z/2)\] on $X$ is induced from the $G$-action on $X$, and $G_{\rho}$ is the isotropy group of $\rho$ under the $G$-action. Note that the two irreducible $\Z/2$-representations  are fixed by the $G$-action and $G_{\rho} = G$ for each $\rho$. Thus, the isomorphism \eqref{decomp:z2:cor} is exactly 
\[K^*_{G}(X)\cong K^*_{Q}(X)\oplus K^{[\tilde{Q}_{sign}]+*}_{Q}(X)\] In each component, the $Q$-action on $X$ is induced from the quotient map $\pi: G\longrightarrow Q$.

\end{proof}

Let
\[ 1 \longrightarrow \Z/2  \buildrel{l}\over\longrightarrow G \buildrel{\pi}\over\longrightarrow Q \longrightarrow 1 \] be a short exact sequence of compact groups and $l(A)$ is contained in the center of $G$.
For any torsion element $\alpha$ in $G$, we have the short exact sequence
\[0\longrightarrow \mathbb{Z}/2\buildrel{i}\over\longrightarrow \Lambda_{G}(\alpha)\buildrel{[\pi, id]}\over\longrightarrow \Lambda_{Q}(\pi(\alpha))\longrightarrow 0\] with \[i(\Z/2)=\{[\beta, 0] \in \Lambda_G(\alpha) \mid \beta\in l(\Z/2)\}  \] contained in the center of $\Lambda_{G}(\pi(\alpha))$. 
In addition, $X^\alpha$ is a $ \Lambda_{G}(\alpha)$-space with the action  by 
$i(\Z/2)$ trivial.

Especially, if $\alpha$ is the nontrivial element in $l(\Z/2)$, then $\pi(\alpha) = 1$ and we have \[\Lambda_{Q}(\pi(\alpha))\cong Q\times \T; \quad
\widetilde{\Lambda_{Q}(\pi(\alpha))}_{\rho}  \cong  \tilde{Q}_{\rho}\times \T. \]
In this case, the central extension  \[ \xymatrix{1\ar[r] & \T \ar[r] & \widetilde{\Lambda_{Q}(\pi(\alpha))}_{\rho}  \ar[r]
& \Lambda_{Q}(\pi(\alpha)) \ar[r] &1} \] is completely determined by \[ \xymatrix{1\ar[r] & \T \ar[r]^{i_Q} & \tilde{Q}_{\rho}  \ar[r]^{p_Q}
& Q\ar[r] &1}, \] thus, by the 3-cocycle $[\tilde{Q}_{\rho}]$.

Then we can get a corollary of Lemma \ref{dcl}.
\begin{lemma} Let
\[ 1 \longrightarrow \Z/2  \buildrel{l}\over\longrightarrow G \buildrel{\pi}\over\longrightarrow Q \longrightarrow 1 \] be a short exact sequence of compact groups and  $l(A)$ is contained in the center of $G$. Let $X$ be a $G$-space with $l(\Z/2)$ acting on
it trivially. For any torsion element $\alpha$ in $G$,
we have the isomorphism
\[ K^*_{\Lambda_G(\alpha)}(X^{\alpha}) \cong K^*_{\Lambda_{Q}(\pi(\alpha))}(X^{\alpha})\oplus K^{[\widetilde{\Lambda_{Q}(\pi(\alpha))}_{sign}]+*}_{\Lambda_{Q}(\pi(\alpha))}(X^{\alpha}).\]

Especially,  if $\alpha$ is the nontrivial element in $l(\Z/2)$, \[K^*_{\Lambda_G(\alpha)}(X^{\alpha}) \cong K^*_{Q}(X^{\alpha})\otimes \Z[q^{\pm}]\oplus K^{[\tilde{Q}_{sign}]+*}_{Q}(X^{\alpha}) \otimes \Z[q^{\pm}]. \]

\label{dcld2}
\end{lemma}

\section{An application of Real Mackey-type decomposition } \label{Real:decomp:dcl}

In this section we give a corollary of \cite[Theorem 1.10]{huanyoung2022}, which is a Real generalization of the Mackey-type decomposition of complex $K$-theory \cite[\S 5]{freed2011b} and, when it is specialized to the complex case, we get \cite[Theorem 3.6, Corollary 3.7]{ngel2017EquivariantCB}. And then we apply it in the computation of Real quasi-elliptic cohomology of $4$-spheres.

First we recall the setting of the theorem. 
Let
\begin{equation}\label{eq:gradedSESLift} 1 \longrightarrow H \buildrel{l}\over\longrightarrow \hat{G} \buildrel{p}\over\longrightarrow \hat{Q} \longrightarrow 1 \end{equation} be an exact sequence of $\Z/2$-graded compact Lie groups where $\hat{Q}$ is nontrivially graded. The ungraded groups of $\hat{G}$ and $\hat{Q}$ are denoted by $G$ and $Q$ respectively.
Given $\epsilon \in \Z/2$ and a complex vector space $V$, write
\begin{equation}\label{eq:conjNotation}
{^{\epsilon}}V
=
\begin{cases}
V & \mbox{if } \epsilon =1, \\
\overline{V} & \mbox{if } \epsilon = -1,
\end{cases}
\end{equation}
where $\overline{V}$ is the complex conjugate vector space of $V$.

The group $\hat{G}$ acts on the set $\Irr (H)$ of isomorphism classes of irreducible unitary representations of $H$: for an irreducible   $H$-representation $\rho_V$ and $\omega \in \hat{G}$,  $\omega \cdot \rho_V$ is defined by
\[
(\omega \cdot \rho_V)(h) = \rho_{{^{\pi(\omega)}}V}( \omega^{-1} h \omega),
\qquad \text{ for any }
h \in H.
\]
For any $x \in H$, the map $\rho_V \rightarrow x \cdot \rho_V$ is an $H$-equivariant isometry. In particular, $H$ acts trivially on $\Irr(H)$ and there is an induced action of $\hat{Q}$ on $\Irr(H)$. 

Fix a representative $V$ of each $[V] \in \Irr(H)$.  By Schur's Lemma, for any representative $W$ of $\omega \cdot [V]$, \[L_{[V],\omega} := \hom_{H}(W, \omega \cdot V)\] is a hermitian line. Following \cite[Section 9.4]{freed2013b}, the composition maps
\begin{equation}
\label{eq:nuComposition}
L_{\omega_1 \cdot [V],\omega_2} \otimes {^{\pi(\omega_2)}}L_{[V], \omega_1}
\longrightarrow
L_{[V], \omega_2 \omega_1},
\qquad
f_2 \otimes f_1
\mapsto
 (\omega_2 \cdot f_1) \circ f_2
\end{equation}
define a $\pi$-twisted extension of $\Irr (H) \git \hat{G}$. For $q \in \hat{Q}$, let \[\mathbb{L}_{[V],q}\]  be the set of all sections $s$ of
\[
\bigcup_{ \omega \in p^{-1}(q)} L_{[V],\omega} \longrightarrow p^{-1}(q) \subset \hat{G}
\]
such that the image of $\rho_W(h) \otimes s(\omega)$ under \eqref{eq:nuComposition} is $s(h\omega)$ for all $h \in H$, where $W$ is the representative of $q \cdot V$. Exactness of the sequence \eqref{eq:gradedSESLift} implies that $\mathbb{L}_{[V],q}$ is one dimensional. The maps \eqref{eq:nuComposition} induce on $\{\mathbb{L}_{[V],q}\}_{[V],q}$ the structure of a $\pi$-twisted extension of $\Irr(H) \git \hat{Q}$, which we denote by \[\hat{\nu}_{\hat{G}}.\] Then we have the decomposition formula.

\begin{theorem}
\label{thm:KRMackeyDecomp}
Let $1 \rightarrow H \rightarrow \hat{G} \rightarrow \hat{Q} \rightarrow 1$ be an exact sequence of $\Z/2$-graded compact Lie groups with $\hat{Q}$ non-trivially graded. Let $\hat{G}$ act on a compact Hausdorff space $X$ with contractible local slices\footnote{\emph{Existence of contractible local slices} means that each $x \in X$ admits a closed $\hat{G}$-stable neighbourhood of the form $\hat{G} \times_{\Stab_{\hat{G}}(x)} S_x$ for a slice $S_x$ which is $\Stab_{\hat{G}}(x)$-equivariantly contractible.} such that $H$ acts trivially. Then there is an isomorphism
\[
KR^{\ast}_{G}(X)
\cong
KR^{\ast +\hat{\nu}_{\hat{G}}}_{Q, \cpt}(X \times \Irr(H)),
\]
where $\hat{Q}$ acts diagonally on $X \times \Irr(H)$, the pullback of $\hat{\nu}_{\hat{G}}$ along $(X \times \Irr(H)) \git \hat{Q} \longrightarrow \Irr (H) \git \hat{Q}$ is again denoted by $\hat{\nu}_{\hat{G}}$ and $KR_{\cpt}(-)$ is $KR$-theory with compact supports.
\end{theorem}

We refer the readers \cite[Section 1.5]{huanyoung2022} for the proof of the theorem and more details.

We are especially in the case when $H$ is $\Z/2$.       The irreducible unitary representations of $\Z/2$ are $1$ and the sign representation $sign$. They are both of the real type. Thus, $ \hat{G}$ acts trivially on $\Irr(\Z/2)$. So $\hat{G}$ acts trivially on the product $S^0\times \Irr(\Z/2)$. 
     Thus, $ \Irr (H) \git \hat{Q} = \{1\} \git \hat{Q} \sqcup \{sign\}\git \hat{Q}$. And we use \[\hat{\nu}_{\hat{G}, 1}, \quad \hat{\nu}_{\hat{G}, sign} \] to denote the restriction of $\pi$-twisted extension of $\hat{\nu}_{\hat{G}}$ to the components $\{1\} \git \hat{Q}$ and $ \{sign\}\git \hat{Q}$ respectively.   In addition, $\hat{\nu}_{\hat{G}, 1}$ gives the trivial twist.
     Thus, 
     by Theorem \ref{thm:KRMackeyDecomp}, 
     \[
KR^{\ast}_{G}(S^0)
\cong
KR^{\ast +\hat{\nu}_{\hat{G}}}_{Q}(S^0 \times \Irr(\Z/2)) \cong \prod_1^2 KR^{\ast}_{Q}(\pt)\oplus KR^{\ast +\hat{\nu}_{\hat{G}, sign}}_{Q}(\pt).
\] So we get the corollary below. 

\begin{corollary} \label{RM:Z2}
    Let $1 \rightarrow \Z/2 \rightarrow \hat{G} \rightarrow \hat{Q} \rightarrow 1$ be an exact sequence of $\Z/2$-graded compact Lie groups with $\hat{Q}$ non-trivially graded. Let $\hat{G}$ act on $S^0$ trivially. Then we have the isomorphism \[
KR^{\ast}_{G}(S^0)
\cong
\prod_1^2 KR^{\ast}_{Q}(\pt)\oplus KR^{\ast +\hat{\nu}_{\hat{G}, sign}}_{Q}(\pt).
\] 
\end{corollary}

\bibliographystyle{amsalpha}
\bibliography{QEllR_comp.bib}

\end{document}